\documentclass{amsart}
\usepackage[latin1]{inputenc}
\usepackage{graphicx}
\usepackage{epstopdf, epsfig}
\usepackage[percent]{overpic}
\usepackage[latin1]{inputenc}
\usepackage{enumitem}
\usepackage{color}
\usepackage{amssymb}

\newcommand{\sgn}{\mbox{\textnormal{sgn}}}
\newcommand{\abs}[1]{\lvert#1\rvert}
\newcommand{\babs}[1]{\big\lvert#1\big\rvert}

\newcommand{\eps}{\varepsilon}

\newcommand{\R}{\mathbb R}

\newcommand{\dt}{\partial_t}
\newcommand{\dx}{\partial_x}

\newcommand{\dz}{\partial_z}

\newcommand{\bU}{{\bf U}}

\newcommand{\nam}{\nabla^\mu}

\newtheorem{remark}{Remark}
\newtheorem{proposition}{Proposition}

\title{Recovering water wave elevation from pressure measurements }

\author{P. Bonneton
 \and D. Lannes}

\begin{document}

\maketitle

\begin{abstract}
The reconstruction of water wave elevation from bottom pressure measurements is an important issue for coastal applications, but corresponds to a difficult mathematical problem. In this paper we present the derivation of a method which allows the elevation reconstruction of water waves in intermediate and shallow waters.  From comparisons with numerical Euler solutions and wave-tank experiments we show that our nonlinear method provides much better results of the surface elevation reconstruction compared to the linear transfer function approach commonly used in coastal applications. More specifically, our method
accurately reproduces the peaked and skewed shape of nonlinear wave fields. Therefore, it is particularly relevant for applications on extreme waves and wave-induced sediment transport. 
\end{abstract}

%
%
\section{Introduction}
 
Accurate measurements of surface waves in the coastal zone are crucial for many applications, such as coastal flooding, navigation and offshore platform safety or wave-induced circulation and sediment transport. Underwater pressure transducers have long been used for measuring surface waves. The reason is that these wave gauges are cheap, easy to deploy at the sea bottom and are much less affected by storms, ships and vandalism than surface wave buoys (\cite{kennedy2010}). However, the reconstruction of the wave field from bottom pressure measurements is a difficult mathematical problem. 

The hydrostatic assumption is, most of the time, relevant for describing long waves, such as tsunamis and tides. However, as long waves propagate shoreward, nonlinear interactions are enhanced by the water depth decrease and can lead to the formation of dispersive shocks (e.g. \cite{madsen2008}, \cite{tissier2011} and \cite{Bonneton2015}). In that case the hydrostatic assumption is no longer valid (see \cite{martins2017}). 

For wind-generated waves the commonly used practice is to recover the wave field by means of a transfer function based on linear wave theory (e.g. \cite{guza1980}, \cite{bishop1987} and \cite{tsai2005}). This method allows, in intermediate water depth, a satisfactory estimate of bulk wave parameters, such as the significant wave height (\cite{tsai2005}).  In shallow water the nonlinearity effects increase and contribute to the peaked and skewed shape of waves. A correct description of these wave properties is of paramount importance for many coastal applications. For instance, studies on wave submersion require an accurate characterization of the highest incoming wave crests. Furthermore, the wave asymmetry and skewness play an important role in wave-induced sediment transport (e.g. \cite{dubarbier2015}).   For all these applications, the transfer function is no longer suitable and a nonlinear reconstruction method is required.

In recent years, several studies have been devoted to the nonlinear reconstruction of water-wave profile from pressure measurements. For steady one-dimensional water waves traveling at constant celerity, \cite{deconinck2012} and \cite{oliveras2012} derived, from the Euler equations, a nonlinear non-local implicit relationship between the pressure and the surface elevation. \cite{constantin2012} obtained, for solitary waves, an explicit formula relating the pressure and the wave elevation. For periodic waves, reconstruction methods were derived which required either solving an ordinary differential equation (\cite{clamond2013a}) or solving an implicit functional equation (\cite{clamond2013b}). All these recent nonlinear recovery methods hold only for steady waves propagating at a constant celerity. Therefore, they cannot be directly applied to real ocean surface waves which are inherently non-stationary and random. However, from their nonlinear constant-celerity approach, \cite{oliveras2012} obtained a heuristic approximation which can be applied to waves that are not necessarily traveling with constant celerity.

In this paper we present the derivation of nonlinear formulas which allow the elevation reconstruction of real surface waves in intermediate and shallow waters.  After presenting the modelling framework in section 2, we derive a weakly-nonlinear fully-dispersive reconstruction formula in section 3, which writes:
$$
\zeta_{\rm NL}=\zeta_{\rm L}-\frac{1}{g} \dt \big(\zeta_{\rm L} \dt \zeta_{\rm L}\big),
$$
where $\zeta_{\rm L}$ and $\zeta_{\rm NL}$ are the linear and nonlinear elevation approximations respectively and $g$ the acceleration of gravity. We discuss in section 4 how to apply this nonlinear method for practical applications where the only input data are bottom pressure time series recorded at a given measurement point. In section 5 we show
that this simple and easy-to-use nonlinear formula  provides much better reconstructions of the surface elevation compared to the classical transfer function approach, in particular in terms of maximum wave elevation and wave skewness.

%
%
\section{Modelling framework}

%
%
\subsection{Notations}\label{sectnot}
 We denote by $z$ the vertical variable and by $X\in \R^d$ the horizontal variables, with  $d$ the surface dimension ($d=1$ or 2) . $\nabla$ and $\Delta$ are the gradient and Laplace operators with respect to the horizontal variables, and $\nabla_{X,z}$ and $\Delta_{X,z}$ are their three-dimensional counterparts.

We denote by  $\widehat{\cdot}$ the Fourier transform in space, and by $\widetilde{\cdot}$ or ${\mathcal F}_t$ the Fourier transform in time, so that for a function of space and time, one has
$$
\widehat{u}(t,\xi)=\frac{1}{(2\pi)^{d/2}}\int_{\R^d} e^{-ix\cdot \xi}u(t,X)dX
\quad\mbox{ and }\quad
\widetilde{u}(\omega,X)=\frac{1}{(2\pi)^{1/2}}\int_{\R} e^{-i\omega t}u(t,X)dt.
$$

We also denote by $f(D)$ Fourier multipliers in space, and $g(D_t)$ Fourier multipliers in time, defined as
$$
\widehat{f(D)u(t,\cdot)}(\xi)=f(\xi)\widehat{u}(t,\xi)
\quad\mbox{ and }\quad
\widetilde{g(D_t)u(\cdot,X)}(\omega)=g(\omega)\widetilde{u}(\omega,X).
$$

%
%
\subsection{Physical background}
\label{sectphy}

We consider three-dimensional waves propagating in intermediate and shallow water depths. We denote $z=\zeta(t,X)$ the elevation of the free surface above the still water level $z=0$, and by $z=-h_b(X)$ the bottom elevation. We are looking for a relationship between pressure time series measured at the bottom, $P_{\rm b}(t,X_0)$, and the elevation $\zeta(t,X_0)$ at the same horizontal location $X_0$.

 In most coastal environments, the bottom elevation is a slowly varying function of $X$. However, for wave modelling over large coastal areas the bottom variation can not be neglected. By contrast, for our local reconstruction approach it is justified to neglect the bottom variation (see \S \ref{sectnonflat}). The classical transfer function method, widely used in coastal engineering, also relies on this assumption. In the following, the bottom elevation is given by  $z=-h_0$, where $h_0$ is constant. Our approach cannot be applied to strongly varying bottoms like those related to coastal structures, except if the pressure sensor is located several wavelength offshore the structure. 

The presence of a background current, defined as the mean current in the frame of the seabed, can affect the propagation of waves in the coastal zone. For instance, the waves can encounter significant currents close to river mouths or tidal inlets. However, in wave-dominated environments the background current is usually much smaller than the phase velocity, and its effects can be neglected (see appendix \ref{app_current1}). This is the assumption we make in the core of the paper, but we also present in appendix \ref{app_current} an attempt to generalize our reconstruction method in the presence of a vertically-uniform horizontal background current.

\subsection{The equations of motion}
We consider here an incompressible homogeneous inviscid fluid delimited above by a free surface and below by a flat bottom. 
Assuming that the flow is irrotational, the velocity field $\bU$ of the fluid is given by $\bU=\nabla_{X,z}\Phi$, where the velocity potential $\Phi$ satisfies the mass conservation equation
\begin{equation}
\label{Laplace}\Delta \Phi +\dz^2 \Phi=0\quad\mbox{ in } \Omega(t),
\end{equation}
where $\Omega(t)$ is the fluid domain at time $t$ and is given by
$$
\Omega(t)=\{(X,z)\in \R^{d+1}, -h_0<z<\zeta(t,X)\}.
$$
The fluid motion is governed by Euler's equation, or equivalently Bernoulli's equation when written in terms of $\Phi$,
\begin{equation}\label{Bernoulli}
\dt \Phi +g z+\frac{1}{2}\abs{\nabla\Phi}^2+\frac{1}{2}\abs{\dz\Phi}^2=-\frac{1}{\rho}(P-P_{\rm atm}),
\end{equation}
where $\rho$ is the density of the fluid and $P_{\rm atm}$ the (constant) atmospheric pressure. These equations are complemented by boundary conditions. At the bottom we have
\begin{equation}
\label{BC_fond}\dz \Phi=0 \quad\mbox{ on }z=-h_0;
\end{equation}
at the surface, we have the classical kinematic equation on $\zeta$,
\begin{equation}\label{kinematic}
\dt \zeta =\dz \Phi -\nabla\zeta \cdot \nabla \Phi \quad \mbox{ on }\quad z=\zeta ,
\end{equation}
and the pressure continuity,
\begin{equation}\label{BC_P}
P=P_{\rm atm} \quad \mbox{ on }\quad z=\zeta .
\end{equation}

%
%
\subsection{Dimensionless equations}

Three main length scales are involved in this problem: the typical horizontal scale $L$, the amplitude $a$ of the wave, and the water depth $h_0$. We shall use several dimensionless numbers formed with these quantities, namely,
$$
\eps=\frac{a}{h_0}, \qquad \mu=\frac{h_0^2}{L^2}, \qquad \sigma=\frac{a}{L}.
$$
These parameters are respectively called {\it nonlinearity}, {\it shallowness} and {\it steepness} parameters and are related through the identity
$$
\sigma=\eps\sqrt{\mu} .
$$
The different variables and functions involved in this problem can be put in dimensionless form using the relations
$$
X'=\frac{X}{L},\quad z'=\frac{z}{h_0}, \quad t'=\frac{\sqrt{gh_0}}{L} t,\quad \zeta'=\frac{\zeta}{a},\quad \Phi'=\frac{h_0}{aL\sqrt{gh_0}}\Phi,\quad P'=\frac{P}{\rho g h_0},
$$
where the primes are used to denote dimensionless quantities. \\
Omitting the primes for the sake of clarity, the fluid domain becomes, in dimensionless form,
$$
\Omega_\eps=\{(X,z)\in \R^{d+1}, -1<z<\eps \zeta(t,X)\},
$$
and the equations \eqref{Laplace}-\eqref{BC_fond} become
\begin{align}
\label{Laplace_ND}\mu\Delta \Phi +\dz^2 \Phi=0&\quad\mbox{ in } \Omega_\eps(t),\\
\label{BC_fond_ND}\dz \Phi=0& \quad\mbox{ on }z=-1.
\end{align}
Similarly, the dimensionless Bernoulli equation is 
\begin{equation}\label{Bernoulli_ND}
\dt \Phi +\frac{1}{\eps}z+\frac{\eps}{2}\abs{\nabla\Phi}^2+\frac{\eps}{2\mu}\abs{\dz\Phi}^2=-\frac{1}{\eps}(P-P_{\rm atm}),
\end{equation}
while for the kinematic equation we have
\begin{equation}\label{kinematic_ND}
\mu \dt \zeta =\dz \Phi -\eps\mu \nabla\zeta \cdot \nabla \Phi \quad \mbox{ on }\quad z=\eps\zeta.
\end{equation}

%
%
\subsection{A general formula for $\zeta$}

Evaluating \eqref{Bernoulli_ND} at the surface $z=\eps\zeta$  and at the bottom $z=-1$ respectively, and using the notations
$$
\psi=\Phi_{\vert_{z=\eps\zeta}}\qquad \quad \Phi_{\rm b}=\Phi_{\vert_{z=-1}},\quad\mbox{ and}\quad P_{\rm b}=P_{\vert_{z=-1}},
$$
we obtain respectively
\begin{align}
\label{eqsurf_ND}
\dt \psi &+\zeta+\frac{\eps}{2}\vert \nabla\psi\vert^2
-\frac{\eps}{2\mu}(1+ \eps^2\mu\vert \nabla\zeta\vert^2)(\dz \Phi_{\vert_{z=\eps\zeta}})^2 =0, \\
\label{eqbott_ND}
\dt \Phi_{\rm b}&-\frac{1}{\eps}+\frac{\eps}{2}\vert \nabla\Phi_{\rm b}\vert^2=-\frac{1}{\eps}(P_{\rm b}-P_{\rm atm}).
\end{align}
From these equations, we obtain the following exact expression for the surface elevation
\begin{equation}\label{formuleexacte_ND}
\zeta=\zeta_{\rm H}+\dt \Phi_{\rm b}-\dt \psi+\frac{\eps}{2}\big(\vert \nabla\Phi_{\rm b}\vert^2-\vert \nabla\psi\vert^2\big)+\frac{\eps}{2\mu} (1+ \eps^2\mu\vert \nabla\zeta\vert^2)(\dz \Phi_{\vert_{z=\eps\zeta}})^2 ,
\end{equation}
where $\zeta_{\rm H}$ is the dimensionless hydrostatic reconstruction
\begin{equation}\label{rec_hydro_ND}
\zeta_{\rm H}=\frac{1}{\eps}(P_{\rm b}-P_{\rm atm}-1).
\end{equation}
 
The formula \eqref{formuleexacte_ND} is exact but involves quantities that cannot be expressed in terms of the pressure measured at the bottom $P_{\rm b}$. Our goal is to derive approximate formulas that can be expressed as a function of the measured quantity $P_{\rm b}$, or equivalently $\zeta_{\rm H}$. In order to do so, we shall perform an asymptotic expansion of \eqref{formuleexacte_ND} in terms of the {\it steepness parameter} $\sigma$, which is a small parameter for most oceanic waves.

In this paper, we consider {\it small steepness} configurations in {\it shallow or intermediate depth}\footnote{It would be possible to generalize the nondimensionalization as presented in \cite{lann_POF2009} in order to also cover the deep water case $\mu \gg 1$, but this is not relevant for the applications we are interested in here.}, that is
\begin{equation}\label{regime}
\sigma=\eps\sqrt{\mu}\ll 1\quad \mbox{ and }\quad \eps,\mu \lesssim 1.
\end{equation}

In particular, this covers the following cases:
\begin{itemize}
\item \ Large amplitude ($\eps\sim 1$) waves in shallow water ($\mu\ll 1$)
\item \ Small amplitude ($\eps\ll 1$) waves in intermediate depth ($\mu\sim1$).
\end{itemize}

%
%
\section{Asymptotic reconstruction formulas}
\label{section_asymptotic}
Our goal in this section is to derive approximate expressions of the exact  formula \eqref{formuleexacte_ND}, in the small steepness regime \eqref{regime}, as a function of the measured quantity $\zeta_{\rm H}$. We also derive simplified expressions in the shallow water case ($\mu\ll 1$).

%
%
\subsection{Linear reconstruction for the surface elevation}\label{sectlin}

As shown in Appendix \ref{appasPhi}, the velocity potential $\Phi$ is given at first order by the linear formula
\begin{equation}\label{app0}
\Phi=\frac{\cosh(\sqrt{\mu}(z+1)\vert D\vert)}{\cosh(\sqrt{\mu}\vert D\vert)}\psi+O(\sigma),
\end{equation}
where, we recall that $\psi=\Phi_{\vert_{z=\eps\zeta}}$, and where we used the notations for Fourier multipliers in space introduced in \S \ref{sectnot}.\\
Using the formula \eqref{app0}, it is possible to approximate the various terms on the right-hand-side of \eqref{formuleexacte_ND}. In particular we get that
$$
\frac{\eps}{2}\big(\vert \nabla\Phi_{\rm b}\vert^2-\vert \nabla\psi\vert^2\big)
+ \frac{\eps}{2\mu} (1+\sigma^2 \vert \nabla\zeta\vert^2)(\dz \Phi_{\vert_{z=\eps\zeta}})^2 =O(\sigma)
$$
so that \eqref{formuleexacte_ND} gives
\begin{align*}
\zeta&=\big(\frac{1}{\cosh(\sqrt{\mu}\vert D\vert)}-1\big)\dt\psi+\zeta_{\rm H}+O(\sigma)\\
&=\big(1-\frac{1}{\cosh(\sqrt{\mu}\vert D\vert)}\big)\zeta+\zeta_{\rm H}+O(\sigma),
\end{align*}
where we used \eqref{eqsurf_ND} to derive the second identity.
We then have
$$
\frac{1}{\cosh(\sqrt{\mu}\vert D\vert)} \zeta=\zeta_{\rm H}+O(\sigma).
$$
Neglecting the $O(\sigma)$ terms, we obtain the following linear reconstruction formula
\begin{equation}\label{rec_lin0}
\zeta_{\rm L}(t,X)=\big[\cosh(\sqrt{\mu}\vert D\vert)\zeta_{\rm H}(t,\cdot)\big](X).
\end{equation}
A generalization of this equation, when the pressure is measured at some point located above the bottom, is also given in \eqref{rec_lin0_meas} in Appendix \ref{appnotbot}.
We show in the next section how to make this formula more precise by including quadratic nonlinear terms.

%
%
\subsection{Quadratic reconstruction for the surface elevation}\label{sectquadr}

In order to include nonlinear corrections to the linear reconstruction formula \eqref{rec_lin0}, we need a quadratic approximation of the velocity potential $\Phi$. As shown in Appendix \ref{appasPhi}, a second order approximation of the velocity potential $\Phi$ is given by the following formula, which is quadratic in $(\zeta,\psi)$,
\begin{equation}\label{app1}
\Phi=\frac{\cosh(\sqrt{\mu}(z+1)\abs{D})}{\cosh(\sqrt{\mu}\abs{D})}\big(\psi-\eps \zeta G_0\psi\big) +O(\sigma^2),
\end{equation}
with 
$$
G_0=\sqrt{\mu}\abs{D}\tanh(\sqrt{\mu}\abs{D}).
$$
We use this approximation to derive a second order approximation of the formula \eqref{formuleexacte_ND} for $\zeta$. We show below how to approximate the different components of \eqref{formuleexacte_ND}:
$$
\zeta=\zeta_{\rm H}+A+B+C,
$$
where, $A:=\dt \Phi_{\rm b}-\dt \psi$, $B:=\frac{\eps}{2}\big(\vert \nabla\Phi_{\rm b}\vert^2-\vert \nabla\psi\vert^2\big)$ and $ C:=\frac{\eps}{2\mu} (1+\eps^2\mu \vert \nabla\zeta\vert^2)(\dz \Phi_{\vert_{z=\eps\zeta}})^2$.
\begin{itemize}
\item \ Approximation of $A$. From the second order approximation of $\Phi$ given in \eqref{app1}, one has
$$
A=\big(\frac{1}{\cosh(\sqrt{\mu}\abs{D})}-1\big)\dt \psi-\eps \frac{1}{\cosh(\sqrt{\mu}\abs{D})} \dt(\zeta G_0 \psi)+O(\sigma^2).
$$
Plugging \eqref{app1} into the kinematic equation \eqref{kinematic_ND}, one also gets
$$
G_0\psi=\mu \dt \zeta  +O(\sqrt{\mu}\sigma)
$$
so that
$$
A=\big(\frac{1}{\cosh(\sqrt{\mu}\abs{D})}-1\big)\dt \psi-\eps \mu \frac{1}{\cosh(\sqrt{\mu}\abs{D})} \dt(\zeta \dt \zeta)+O(\sigma^2).
$$
\item \ Approximation of $B$. Using the first order approximation \eqref{app0} of $\Phi$, one gets
\begin{align*}
B&=\frac{\eps}{2}\babs{\frac{1}{\cosh(\sqrt{\mu}\abs{D})}\nabla\psi}^2-\frac{\eps}{2}\abs{\nabla\psi}^2
+O(\sigma^2)
\end{align*}
\item \ Approximation of $ C$. Let us first write
\begin{align*}
 C&:= \big(1-\frac{1}{\cosh(\sqrt{\mu}\abs{D})})C+\frac{1}{\cosh(\sqrt{\mu}\abs{D})} C\\
 &=\big(1-\frac{1}{\cosh(\sqrt{\mu}\abs{D})})C+\frac{\eps}{2\mu} \frac{1}{\cosh(\sqrt{\mu}\abs{D})}(G_0\psi)^2+O(\sigma^2),
 \end{align*}
 where we used the first order approximation \eqref{app0} of $\Phi$ to derive the second identity. Approximating as above $G_0\psi$ by $\mu\dt \zeta$, this yields
 $$
 C=\big(1-\frac{1}{\cosh(\sqrt{\mu}\abs{D})})C+\frac{1}{2} \eps\mu \frac{1}{\cosh(\sqrt{\mu}\abs{D})}(\dt\zeta)^2+O(\sigma^2).
 $$
\end{itemize}
We deduce from the lines above that
\begin{align*}
\zeta=&\zeta_{\rm H}+\big(\frac{1}{\cosh(\sqrt{\mu}\abs{D})}-1\big)\big[\dt \psi-C\big]\\
&+\eps \mu \frac{1}{\cosh(\sqrt{\mu}\abs{D})} \big(\frac{1}{2}(\dt \zeta)^2-\dt(\zeta \dt \zeta)\big) + \frac{\eps}{2}\babs{\frac{1}{\cosh(\sqrt{\mu}\abs{D})}\nabla\psi}^2-\frac{\eps}{2}\abs{\nabla\psi}^2+O(\sigma^2).
\end{align*}
We now remark further that \eqref{eqsurf_ND} can be written
$$
\dt \psi-C=-\zeta-\frac{\eps}{2}\abs{\nabla\psi}^2,
$$
so that we finally get
\begin{align*}
\frac{1}{\cosh(\sqrt{\mu}\abs{D})}\zeta=\zeta_{\rm H}-\eps \frac{1}{\cosh(\sqrt{\mu}\abs{D})}&\big[
\frac{1}{2}\abs{\nabla\psi}^2+\mu\dt (\zeta\dt \zeta)-\frac{1}{2}\mu (\dt\zeta)^2\\
&-\frac{1}{2}\cosh(\sqrt{\mu}\abs{D})\babs{\frac{1}{\cosh(\sqrt{\mu}\abs{D})}\nabla\psi}^2\big]+O(\sigma^2).
\end{align*}
Neglecting the $O(\sigma^2)$ terms and multiplying by $\cosh(\sqrt{\mu}\abs{D})$, we obtain
\begin{align}
\nonumber
\zeta=\cosh(\sqrt{\mu}\abs{D})\zeta_{\rm H}-\eps&\big[
\frac{1}{2}\abs{\nabla\psi}^2+\mu\dt (\zeta\dt \zeta)-\frac{1}{2}\mu (\dt\zeta)^2\\
\label{rec_quad}
&-\frac{1}{2}\cosh(\sqrt{\mu}\abs{D})\babs{\frac{1}{\cosh(\sqrt{\mu}\abs{D})}\nabla\psi}^2\big].
\end{align}
In order to simplify the nonlinear terms, we need now to assume that the horizontal dimension is equal to one ($d=1$); one then has $\cosh(\sqrt{\mu}\abs{D})=\cosh(\sqrt{\mu}D)$.  
Using the trigonometric formula $\cosh(a+b)=\cosh(a)\cosh(b)+\sinh(a)\sinh(b)$, one readily gets the following identity
$$
\cosh(\sqrt{\mu} D)(fg)=(\cosh(\sqrt{\mu} D)f)(\cosh(\sqrt{\mu} D)g)+(\sinh(\sqrt{\mu} D)f)(\sinh(\sqrt{\mu} D)g),
$$
from which one deduces that 
\begin{align*}
\cosh(\sqrt{\mu}\abs{D})\babs{\frac{1}{\cosh(\sqrt{\mu}\abs{D})}\dx\psi}^2&=(\dx\psi)^2+(\tanh(\sqrt{\mu}D)\dx \psi)^2\\
&=(\dx\psi)^2-\frac{1}{\mu}(G_0 \psi)^2.
\end{align*}
Approximating as above $G_0\psi$ by $\mu \dt \zeta$, we deduce from \eqref{rec_quad} that
$$
\zeta=\cosh(\sqrt{\mu}\abs{D})\zeta_{\rm H}-\eps\mu \dt \big(\zeta \dt \zeta\big),
$$
and finally, up to $O(\eps^2\mu^2)$ terms,
\begin{equation}\label{rec_quad0}
\zeta_{\rm NL}=\zeta_{\rm L}-\sqrt{\mu}\sigma \dt \big(\zeta_{\rm L} \dt \zeta_{\rm L}\big)\qquad (\mbox{horizontal dimension }d=1)
\end{equation}
with  $\zeta_{\rm L}$ given by linear formula \eqref{rec_lin0}.
A generalization of this formula when the pressure is measured at some distance above the bottom is provided by \eqref{notbotNL} in Appendix \ref{appnotbot}.
The new nonlinear reconstruction formula \eqref{rec_quad0} represents the main result of this paper. This equation can be rewritten:
$$
\zeta_{\rm NL}=\zeta_{\rm L}-\sqrt{\mu}\sigma \zeta_{\rm L} \dt^2 \zeta_{\rm L}-\sqrt{\mu}\sigma  \big(\dt \zeta_{\rm L}\big)^2.
$$ 

The first nonlinear term on the right-hand side mainly contributes at the wave extrema, by reducing the wave troughs and amplifying the wave crests. The second nonlinear term strengthens the wave skewness and asymmetry. In comparison with the linear reconstruction \eqref{rec_lin0}, the nonlinear one leads to more peaked wave crests and flatter troughs in agreement with wave observations (see section \ref{validation}).

\cite{oliveras2012} also derived nonlinear reconstruction formulas but under the more restrictive assumption that waves are steady and propagate at a constant celerity. However, they also obtained a heuristic formula which can be applied to a wider range of applications. The authors noted: {\it This formula is obtained somewhat heuristically, and its justification rests on the fact that it agrees extremely well with both numerical and experimental data}. This nonlinear formula writes
\begin{equation}\label{rec_heur}
\zeta_{HE}=\frac{\zeta_{\rm L}} {1-\sigma D \sinh \big(\sqrt{\mu} D \big)\zeta_{\rm H}}\qquad (\mbox{horizontal dimension }d=1),
\end{equation}
and its performance is compared, in section \ref{validation}, to those of the nonlinear reconstruction formulas derived in the present paper. A generalization of the nonlinear reconstruction \eqref{rec_quad0}, in the presence of a background current, is presented in appendix \ref{app_current1}.
%
%
\subsection{Reconstruction formulas in shallow water}

The shallow water regime ($\mu\ll 1$, $\eps=O(1)$) is a particular case of the small steepness regime ($\sigma=\eps\sqrt{\mu}\ll 1$) considered above. In this particular case, it is possible to derive simpler reconstruction formulas by making a Taylor expansion with respect to $\mu$ of the linear and nonlinear formulas \eqref{rec_lin0} and \eqref{rec_quad0}. \\
Recognizing that 
$$
\cosh(\sqrt{\mu}|\xi|)=1+\frac{\mu}{2}|\xi|^2+O(\mu^2),
$$
we obtain the following simplification of the linear reconstruction formula \eqref{rec_lin1},
\begin{equation}\label{rec_linSW0}
\zeta_{\rm SL}=\zeta_{\rm H}-\frac{\mu}{2}\Delta \zeta_{\rm H}\qquad (\mbox{horizontal dimension }d=1,2);
\end{equation}
we refer to \eqref{rec_linSW0_meas} in Appendix \ref{appnotbot} for a generalization of this formula when the pressure is measured at some point located above the bottom.
The nonlinear reconstruction formula \eqref{rec_quad0} gives similarly in shallow water
\begin{equation}\label{rec_quadSW}
\zeta_{\rm SNL}=\zeta_{\rm SL}-\eps\mu \dt \big(\zeta_{\rm SL} \dt \zeta_{\rm SL}\big)\qquad (\mbox{horizontal dimension }d=1,2)
\end{equation}
(the formula \eqref{rec_quad0} has been established only in dimension $d=1$, but \eqref{rec_quadSW} can easily be derived also in dimension $d=2$ from \eqref{rec_quad} in the shallow water regime). We refer to \eqref{rec_quadSWab} for a generalization of this formula when the pressure is measured above the bottom.

\subsection{A word on the flat bottom assumption}\label{sectnonflat}
Assume that the bottom is given in dimensionless variables by $z=-1+\beta b$ for some function $b$ that vanishes at the measurement point (i.e. the reference depth $h_0$ is the depth at rest at the measurement point), and where $\beta$ as well as the bottom steepness $\sigma_{\rm b}$ are given by
$$
\beta=\frac{a_{\rm b}}{h_0},\qquad \sigma_{\rm b}=\frac{a_{\rm b}}{L},
$$
the lenghth $a_{\rm b}$ being the scale of the amplitude of the bottom variations. With $\Phi_{\rm b}$, $P_{\rm b}$, and $\zeta_{\rm H}$ now given by
$$
\Phi_{\rm b}=\Phi_{\vert_{z=-1+\beta b}},\qquad P_{\rm b}=P_{\vert_{z=-1+\beta b}} \quad \mbox{ and }\quad
\zeta_{\rm H}=\frac{1}{\eps}(P_{\rm b}-P_{\rm atm}-1+\beta b),
$$
one readily checks that the formula \eqref{formuleexacte_ND} remains valid up to $O(\eps \sigma_{\rm b}^2)$ terms.\\
Taking into account the bottom contribution, the quadratic formula \eqref{app1} for the velocity potential becomes, denoting by $\Phi_{\rm flat}$ the formula \eqref{app1},
$$
\Phi=\Phi_{\rm flat} +\sigma_{\rm b} \frac{\sinh(\sqrt{\mu}z \vert D\vert)}{\cosh(\sqrt{\mu} \vert D\vert)}  \frac{\nabla}{\vert D\vert }\cdot   (b \frac{1}{\cosh(\sqrt{\mu}\vert D\vert)}\nabla \psi)+O(\sigma^2,\sigma_{\rm b}^2)
$$
(see \cite{lann_POF2009}). Proceeding as in \S \ref{sectquadr}, the reconstruction formula \eqref{rec_quad0} becomes, if we include the bottom contribution,
$$
\zeta_{\rm NL}=\zeta_{\rm L}-\sqrt{\mu}\sigma \dt \big(\zeta_{\rm L} \dt \zeta_{\rm L}\big)+\sqrt{\mu}\sigma_{\rm b} \dx b \big( \frac{1}{\cosh(\sqrt{\mu}\vert D\vert)}\dx \zeta\big)+O(\mu\sigma^2,\mu\sigma_{\rm b}^2)
$$
(where we used the fact that $b$ vanishes at the measurement point); note that in shallow water, the bottom contribution simplifies into $\sqrt{\mu}\sigma_{\rm b}\dx b\dx \zeta$. \\  The parameter $\sqrt{\mu}\sigma_{\rm b}$ being very small for many coastal applications we neglect throughout this paper the bottom contribution; this contribution could be taken into account for right-going waves using Proposition \ref{prop2}.

%
%
\section{Reconstruction formulas for practical applications}\label{sectpractical}
\label{practical}

We have derived in the preceding section new nonlinear formulas to reconstruct the water elevation from pressure measurements. These formulas involve a Fourier multiplier in space which requires the knowledge of $\zeta_{\rm H}$ (or equivalently $P_{\rm b}$) over the whole horizontal space $\R^2$. While, for most ocean applications, $\zeta_{\rm H}$ is only known at one measurement point $\underline{X}$. To overcome this limitation, we show in this section how to replace the Fourier multiplier in space by a Fourier multiplier in time. Two distinct approaches are considered depending on the wave type: nonlinear permanent form waves and irregular weakly nonlinear waves.

%
%
%
\subsection{Nonlinear permanent form wave}
\label{permanent}

Permanent or quasi-permanent form waves (i.e. traveling waves) are rarely observed in the field but they represent an essential toy model for understanding surface wave dynamics. The assumption that the wave field is stationary in a uniformly translating reference frame at velocity $c_p$, significantly reduces the complexity of the nonlinear water wave problem (e.g.  \cite{oliveras2012}, \cite{constantin2012} or \cite{clamond2013a}). Under this assumption it is straightforward  to replace the Fourier multiplier in space by a Fourier multiplier in time in the linear equation \eqref{rec_lin0}:
\begin{equation}\label{rec_lin0_p}
\zeta_{\rm L}(t,X)=\big[\cosh\big( \frac{\sqrt{\mu}D_t}{c_p}\big)\zeta_{\rm H}(\cdot,X)\big](t)\qquad (\mbox{horizontal dimension }d=1).
\end{equation}
The linear shallow water approximation writes
\begin{equation}\label{rec_linSW0_p}
\zeta_{\rm SL}=\zeta_{\rm H}-\frac{\mu}{2c_p^2}\dt^2\zeta_{\rm H}.
\end{equation}
Contrary to equations \eqref{rec_lin1} and \eqref{rec_linSW0}, which apply to two-dimensional wave fields, it is worth noting that these two formulas are restricted to unidirectional traveling waves.

Since the linear approximation is being estimated, we can apply the nonlinear reconstruction formulas derived in the preceding section. This reconstruction method requires knowing both the pressure time series at the measurement point and the wave celerity $c_p$.

%
%

\subsection{Irregular wave}
The permanent wave form assumption used in the preceding section only applies to a limited number of academic wave cases. In the ocean, wind-generated waves (swell and wind sea) are irregular and random and do not propagate at a constant celerity $c_p$. We show in this section how to replace, in the reconstruction formulas, the Fourier multiplier in space by a Fourier multiplier in time for linear or weakly nonlinear irregular wave fields.

\subsubsection{Linear wave}
\label{linearwave}
To replace the Fourier multiplier in space in the linear reconstruction \eqref{rec_lin0} by a Fourier multiplier in time, we use the fact that $\zeta$ is a solution of the water wave equations. Using the linear approximation \eqref{app0} in  \eqref{kinematic_ND} and \eqref{eqsurf_ND}, while dropping the $O(\sigma)$ terms, we see that $(\zeta,\psi)$ solves 
\begin{equation}\label{syslin_app}
\begin{cases}
\dt \zeta-\frac{1}{\sqrt{\mu}}\tanh(\sqrt{\mu}\vert D\vert)\vert D\vert\psi=0,\\
\dt \psi+\zeta=0;
\end{cases}
\end{equation} 
it follows that $\zeta$, and therefore $\zeta_{\rm H}$, satisfies the equation
\begin{equation}\label{eqlin}
\dt^2\zeta_{\rm H}+\frac{1}{\sqrt{\mu}}\tanh(\sqrt{\mu}\vert D\vert)\vert D\vert\zeta_{\rm H}=0.
\end{equation}
We then use the following proposition to take advantage of this equation to replace the Fourier transform in space that appears in \eqref{rec_lin0} by a Fourier transform in time.
\begin{proposition}\label{prop1}
Let $\lambda: \R \to \R$  be a $C^1$ diffeomorphism with inverse $k:=\lambda^{-1}$. If $u$ is a solution to the equation
\begin{equation}\label{lineareq}
\dt^2 u+\lambda(\abs{D})^2 u=0,
\end{equation}
and if $\lambda(\cdot)$ is odd, then, for all Fourier multiplier $f(\abs{D})$ with $f:\R^+\to \R$, one has the relation,
$$
\forall t,X,\qquad \big[f(\abs{D})u(t,\cdot)\big](X)=[f (\abs{k(D_t)})u(\cdot,X)](t).
$$
\end{proposition}
\begin{proof}
Let us define (see Remark \ref{remright} below for the reason of this choice) the function $\lambda^+: \R^2\to \R$ as
$$
\lambda^+(\xi)=\lambda(\abs{\xi}) \quad \mbox{ if }\quad \xi\cdot {\bf e}_1 \geq 0
\quad\mbox{ and }\quad \lambda^+(\xi)=-\lambda(\abs{\xi}) \quad \mbox{ if }\quad \xi\cdot {\bf e}_1 \leq 0,
$$
where $({\bf e}_1,{\bf e}_2)$ is a  unitary basis of the $\xi$-plane. We can rewrite \eqref{lineareq} under the form
$$
\dt^2 u+\lambda^+(D)^2 u=0.
$$
By taking the Fourier transform in space, one can therefore write $u$ under the general form
\begin{equation}\label{expsol}
\widehat{u}(t,\xi)=A_-(\xi)\exp(i\lambda^+ (\xi)t)+A_+(\xi)\exp(-i\lambda^+ (\xi)t),
\end{equation}
with $A_\pm$ determined by the initial conditions. Considering the double Fourier transform in space and time, we have therefore
$$
({\mathcal F}_{t,x}u)(\omega,\xi)=A_-(\xi)\delta_{\omega=\lambda^+(\xi)}+A_+(\xi)\delta_{\omega=-\lambda^+(\xi)}.
$$
The double Fourier transform of $f(\abs{D})u$ is thus given by
$$
({\mathcal F}_{t,x}f(\abs{D})u)(\omega,\xi)=f(\abs{\xi})A_-(\xi)\delta_{\omega=\lambda^+(\xi)}+f(\abs{\xi})A_+(\xi)\delta_{\omega=-\lambda^+(\xi)}.
$$
Now, if $\omega=\pm \lambda^+(\xi)$, one has $\abs{\xi}=\abs{k(\omega)}$, and therefore
$$
({\mathcal F}_{t,x}f(\abs{D})u)(\omega,\xi)=f(\abs{k(\omega)})A_-(\xi)\delta_{\omega=\lambda^+(\xi)}+f(\abs{k(\omega)})A_+(\xi)\delta_{\omega=-\lambda^+(\xi)}.
$$
Inverting the double Fourier transform then yields the result.
\end{proof}
\begin{remark}\label{remright}
Instead of \eqref{expsol}, a more direct representation formula for the solution $u$ of \eqref{lineareq} could have been
\begin{equation}\label{expsol2}
\widehat{u}(t,\xi)=B_-(\xi)\exp(i\lambda(\vert \xi\vert )t)+B_+(\vert \xi\vert )\exp(-i\lambda (\vert \xi\vert)t);
\end{equation}
this decomposition however does not bear any simple physical signification. Take for instance in dimension $d=1$ the case where $\lambda(\vert D\vert)=\vert D\vert$ so that \eqref{lineareq} is simply given by the wave equation $\dt^2 u -\dx^2 u=0$. The formula \eqref{expsol2} proposes a decompositon of the solution as a sum of the solutions of the scalar equations $(\dt - i\vert D\vert)u=0$ and $(\dt +i\vert D\vert)u=0$; the formula \eqref{expsol} provides a much more natural decomposition as a sum of the solutions of the scalar equations $(\dt - \dx)u=0$ and $(\dt +\dx)u=0$, i.e. as a sum of a {\it left-going wave} and of a {\it right-going wave}.\\
This leads us to the following definition:  a solution $u$ of \eqref{lineareq}  is called a \emph{right-going wave} if $A_-\equiv 0$ in the representation formula \eqref{expsol}.
\end{remark}
Since $\zeta_{\rm H}$ satisfies \eqref{eqlin}, we can use the proposition 1 with $\lambda$ given by: 
\begin{equation}\label{formulalambda}
\lambda(r)=\sgn(r)\left(\frac{r\tanh(\sqrt{\mu}r)}{\sqrt{\mu}} \right)^{1/2}.
\end{equation}
This allows us to transform  \eqref{rec_lin0} into the following linear reconstruction formula, 
\begin{equation}\label{rec_lin1}
\zeta_{\rm L}(t,X)=\big[\cosh\big(\sqrt{\mu} k( D_t)\big)\zeta_{\rm H}(\cdot,X)\big](t)\qquad (\mbox{horizontal dimension }d=1,2).
\end{equation}
In the shallow water regime the linear reconstruction writes
\begin{equation}\label{rec_linSW1}
\zeta_{\rm SL}=\zeta_{\rm H}-\frac{\mu}{2}\dt^2 \zeta_{\rm H}\qquad (\mbox{horizontal dimension }d=1,2).
\end{equation}
It is worth noting that these linear reconstructions are valid in horizontal dimension $d=1$ or $2$. The formula \eqref{rec_lin1} corresponds to the well-known {\it transfer function} method, usually written for ocean applications under the form:
$$
{\mathcal F}_t(\zeta_{\rm L})(\omega,X)=\cosh\big(\sqrt{\mu} k(\omega)\big){\mathcal F}_t(\zeta_{\rm H})(\omega,X),
$$
where $k(\omega)$ is given by the dispersion relation:
\begin{equation}\label{dispersion}
\omega^2=\frac{k(\omega)\tanh(\sqrt{\mu}k(\omega))}{\sqrt{\mu}} .
\end{equation} 

%
%
\subsubsection{Weakly nonlinear narrow-band wave }
\label{WNlinearwave}

We consider in this section weakly nonlinear narrow-band waves such as swell propagating in the coastal zone. To reconstruct the surface elevation of such waves we can apply the nonlinear reconstruction \eqref{rec_quad0} by estimating $\zeta_{\rm L}$ with equation \eqref{rec_lin1}, where the Fourier multiplier in space has been replaced  by a Fourier multiplier in time. However, due to nonlinear interactions, narrow-band spectra develop secondary harmonics of the fundamental frequencies. These secondary harmonics are phase locked, or bound, to their parent waves  and travel at a celerity which is much larger than their intrinsic (linear) phase speed. Thus, the linear dispersive relation \eqref{dispersion} strongly overestimates the wavenumber of the harmonics. Consequently, the linear reconstruction, $\widetilde{\zeta}_{\rm L}(\omega,X)=\cosh\big(\sqrt{\mu} k( \omega)\big)\widetilde{\zeta}_{\rm H}(\omega,X)$, strongly overestimates the amplitude of $\widetilde{\zeta}_{\rm L}(\omega,X)$ for the secondary harmonics. It is therefore necessary to introduce a cut-off frequency $f_c$ between the fundamental frequencies, for which we can apply the linear formula, and the high frequency tail:
\begin{eqnarray}
\widetilde{\zeta}_{\rm L}(\omega,X)&=&\cosh\big(\sqrt{\mu} k( \omega)\big)\widetilde{\zeta}_{\rm H}(\omega,X) \qquad \frac{\omega}{2\pi}\le f_c \label{rec_methodLin} \\
\widetilde{\zeta}_{\rm L}(\omega,X)&=&\widetilde{\zeta}_{\rm H}(\omega,X) \qquad \frac{\omega}{2\pi} > f_c . \nonumber
\end{eqnarray}
 Such a frequency-decomposition is relevant for nonlinear wind-generated waves, and especially for swells, but cannot be applied to nonlinear waves for which the temporal spectrum is a continuous function of the frequency (e.g. solitary waves). 
It is worth noting that contrary to what is generally accepted in the literature for swell reconstruction, the need for such a cut-off is mainly due to the wave nonlinearities rather than to pressure  measurement noise. 

The linear reconstruction, $\zeta_{\rm L}$, being estimated we can apply the  formula  \eqref{rec_quad0} to obtain a nonlinear reconstruction of wave elevation, namely,
\begin{equation}
\label{rec_nous}
\zeta_{{\rm NL}} = \zeta_{\rm L}-\sqrt{\mu}\sigma \zeta_{\rm L} \dt^2 \zeta_{\rm L}-\sqrt{\mu}\sigma (\dt \zeta_{\rm L})^2. 
\end{equation}
 The nonlinearities in this formula involve quadratic interactions among the fundamental modes which fill the elevation spectrum  beyond the cut-off frequency. 

The same frequency decomposition can also be applied to the heuristic reconstruction \eqref{rec_heur}. This equation is thus equivalent, at order $O(\sigma)$, to
\begin{equation}
\label{rec_heurc}
\zeta_{{\rm HE}}\simeq \zeta_{\rm L} + \sigma \zeta_{\rm L} D \sinh \big(\sqrt{\mu} D \big)\zeta_{\rm H} = \zeta_{\rm L}-\sqrt{\mu}\sigma \zeta_{\rm L} \dt^2 \zeta_{\rm L}, 
\end{equation}
which is similar to our equation \eqref{rec_nous} if we neglect the nonlinear term $\sqrt{\mu}\sigma\big(\dt \zeta_{\rm L}\big)^2$. As commented in \S \ref{sectbichro}, this term plays an important role to reproduce the wave skewness. A generalization of the nonlinear reconstruction \eqref{rec_nous}, in the presence of a background current, is presented in appendix \ref{app_current2}.

%
%

\subsubsection{Method implementation}
\label{method}
The nonlinear reconstruction \eqref{rec_nous} is very easy to implement. Indeed, it is a straightforward extension of the commonly used linear transfer function method.
\\

\noindent\underline{Transfer function method}
\\
\begin{enumerate}[label=(\alph*)]
\item \ One considers a measured pressure time series, $P_{\rm b}(t,X_0)$, long enough to contain several peak periods of the wave field.  
\item \ The characteristic water depth, $h_0$, corresponds to the mean water depth, which is equal to the time average of the hydrostatic water depth  $$h_{\rm H}(t,X_0)=\frac{P_{\rm b}-P_{\rm atm}}{\rho_0 g}$$  
(this is consistent because the time average of all the approximations derived in Section \ref{practical} have the same time average as $\zeta_{\rm H}$).
\item \ The dimensionless hydrostatic elevation is given by $$\zeta_{\rm H}(t,X_0)=\frac{P_{\rm b}-P_{\rm atm}-1}{\eps} .$$
\item \ The Fourier transform of the hydrostatic elevation, ${\mathcal F}_t (\zeta_{\rm H})(\omega,X_0)$, is computed.
\item \ ${\mathcal F}_t (\zeta_{\rm L})(\omega,X_0)$ is calculated from  \eqref{rec_methodLin}, where $k(\omega)$ is given by the dispersive relation \eqref{dispersion}.
\item \ Finally, the linear elevation reconstruction, $\zeta_{\rm L}$, is obtained from an inverse Fourier transform:  $\zeta_{\rm L}(t,X_0)= {\mathcal F}_t^{-1}\left( {\mathcal F}_t (\zeta_{\rm L}) \right)$.
\end{enumerate}
~\\
\noindent\underline{Nonlinear reconstruction method}

~\\
The only difference with the linear method is in the last step (f). One  computes not only $\zeta_{\rm L}= {\mathcal F}_{t}^{-1}\left( {\mathcal F}_t (\zeta_{\rm L}) \right)$, but also two other inverse Fourier transforms:
\begin{eqnarray*}
\dt \zeta_{\rm L}&=& {\mathcal F}_t^{-1}\left( i\omega {\mathcal F}_t (\zeta_{\rm L}) \right) \\
\dt^2 \zeta_{\rm L}&=& {\mathcal F}_{t}^{-1}\left( -\omega^2 {\mathcal F}_t (\zeta_{\rm L}) \right).
\end{eqnarray*}
Finally, we use these $\zeta_{\rm L}$ time derivatives to compute the nonlinear elevation reconstruction $\zeta_{\rm NL}$ following equation \ref{rec_nous}. 

The nonlinear reconstruction is essentially based on one direct and three inverse Fourier transforms, which makes the method computationally cheap and thus efficient for operational and real time coastal applications.

%
%

\section{Validations}
\label{validation}

In this section we assess the ability of the formulas derived in the preceding section to reconstruct wave elevation. We compare reconstructed surface elevation fields to numerical Euler solutions or wave-tank experiments. For the latter, the pressure measurements were located at some distance $\delta_m$ above the bed. Thus, we apply the generalized reconstructions for pressure measurements at a given $\delta_m$ (see section \ref{appnotbot}). Here we use the reconstruction formulas  in their dimensional form, as listed in appendix \ref{dimensional}.

%
%
\subsection{Solitary wave}
We compare the reconstruction formulas derived for nonlinear permanent form wave (section \ref{permanent}) to solitary wave solutions computed from the full Euler equations (\cite{dutykh2014}).  This solitary-wave test-case is academic but it provides an useful evaluation of the nonlinear performance of reconstruction formulas \eqref{rec_quadSWb} and  \eqref{rec_heurb} (see also \cite{oliveras2012}). Two solitary wave solutions are considered: $\eps_0=0.1$ (weak nonlinearity) and  $\eps_0=0.4$ (significant nonlinearity), where $\eps_0=a_0/h_0$ and $a_0$ is the amplitude of the solitary wave. Both solutions are characterized by small shallowness parameters:  $\mu = 0.068$ and $\mu = 0.25$ respectively. For such weakly dispersive waves it is natural to apply the shallow water reconstructions.

\begin{figure}
\begin{center}
\includegraphics[width=14cm]{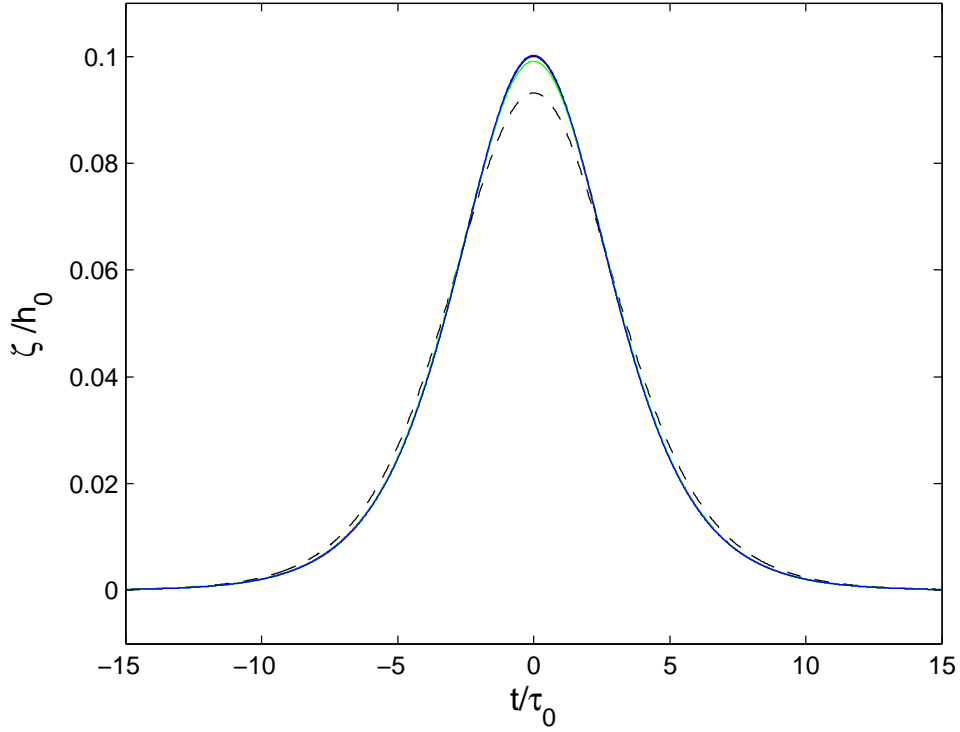}
\caption{Surface elevation reconstruction of a solitary wave, $\eps_0=a_0/h_0=0.1$, $\delta_m=0$. black line: numerical solution of the Euler equations; dashed black line: hydrostatic reconstruction $\zeta_{\rm H}$, equation \eqref{rec_Hyb}; green line: $\zeta_{\rm SL}$, equation \eqref{rec_linSWb}; blue  line: $\zeta_{\rm SNL}$, equation \eqref{rec_quadSWb}. $\tau_0=h_0/c_p$, where $c_p$ is the solitary wave celerity.} 
\label{fig_S1}
\end{center}
\end{figure}

\begin{figure}
\begin{center}
\includegraphics[width=14cm]{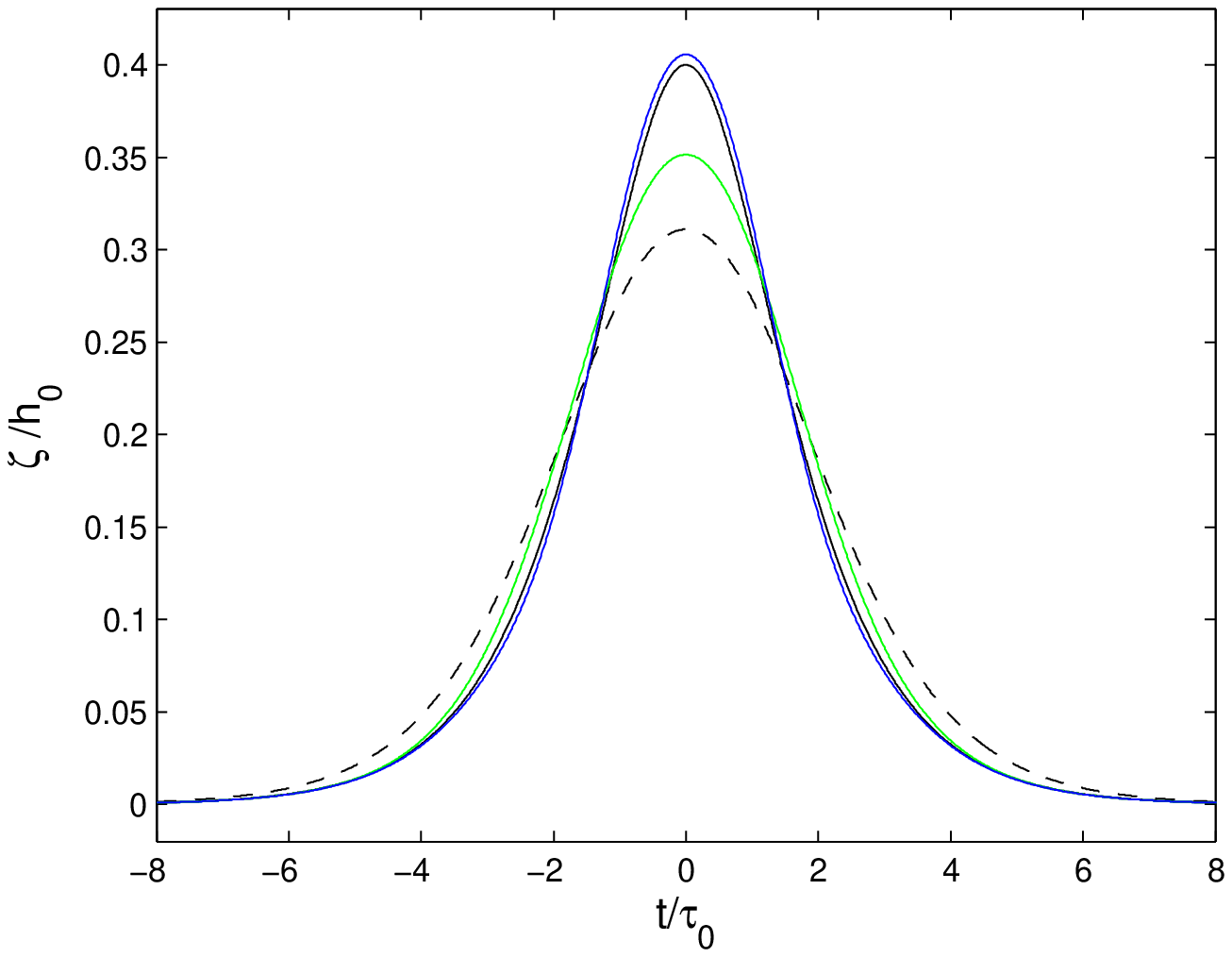}
\caption{Surface elevation reconstruction of a solitary wave, $\eps_0=a_0/h_0=0.4$, $\delta_m=0$. black line: numerical solution of the Euler equations; dashed black line: hydrostatic reconstruction $\zeta_{\rm H}$, equation \eqref{rec_Hyb}; green line: $\zeta_{\rm SL}$, equation \eqref{rec_linSWb}; blue  line: $\zeta_{\rm SNL}$, equation \eqref{rec_quadSWb}. $\tau_0=h_0/c_p$, where $c_p$ is the solitary wave celerity.} 
\label{fig_S2}
\end{center}
\end{figure}

\begin{figure}
\begin{center}
\includegraphics[width=14cm]{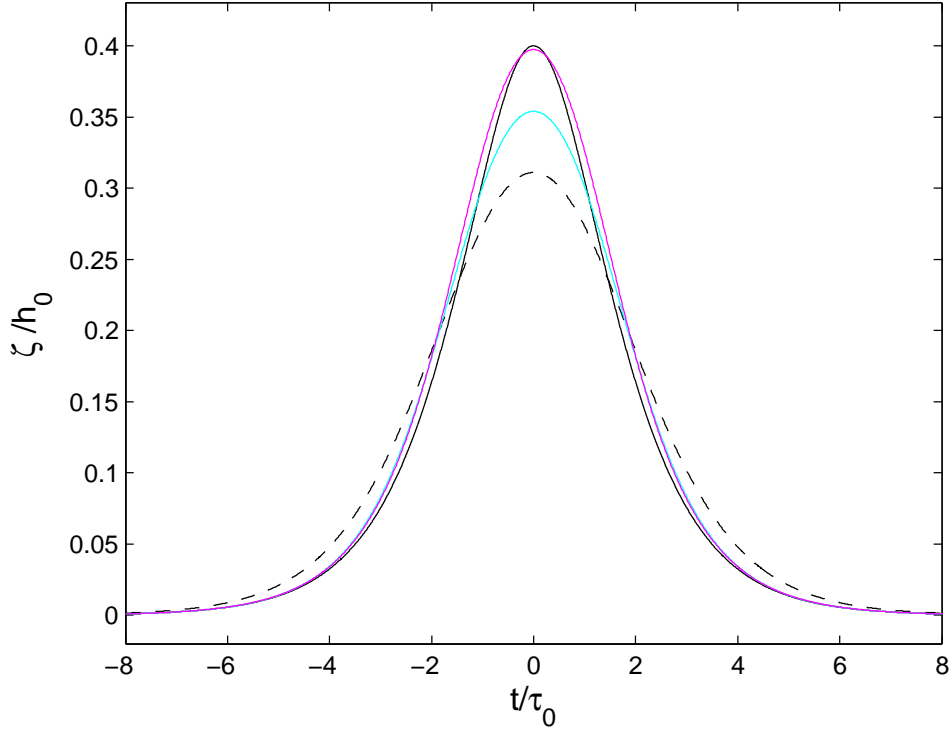}
\caption{Surface elevation reconstruction of a solitary wave, $\eps_0=a_0/h_0=0.4$, $\delta_m=0$. black line: numerical solution of the Euler equations; dashed black line: hydrostatic reconstruction $\zeta_{\rm H}$, equation \eqref{rec_Hyb}; cyan line: $\zeta_{\rm L}$, equation \eqref{rec_lin0_measb}; magenta  line:  heuristic formula, $\zeta_{HE}$, equation \eqref{rec_heurb}. $\tau_0=h_0/c_p$, where $c_p$ is the solitary wave celerity.} 
\label{fig_S3}
\end{center}
\end{figure}

\begin{figure}
\begin{center}
\includegraphics[width=14cm]{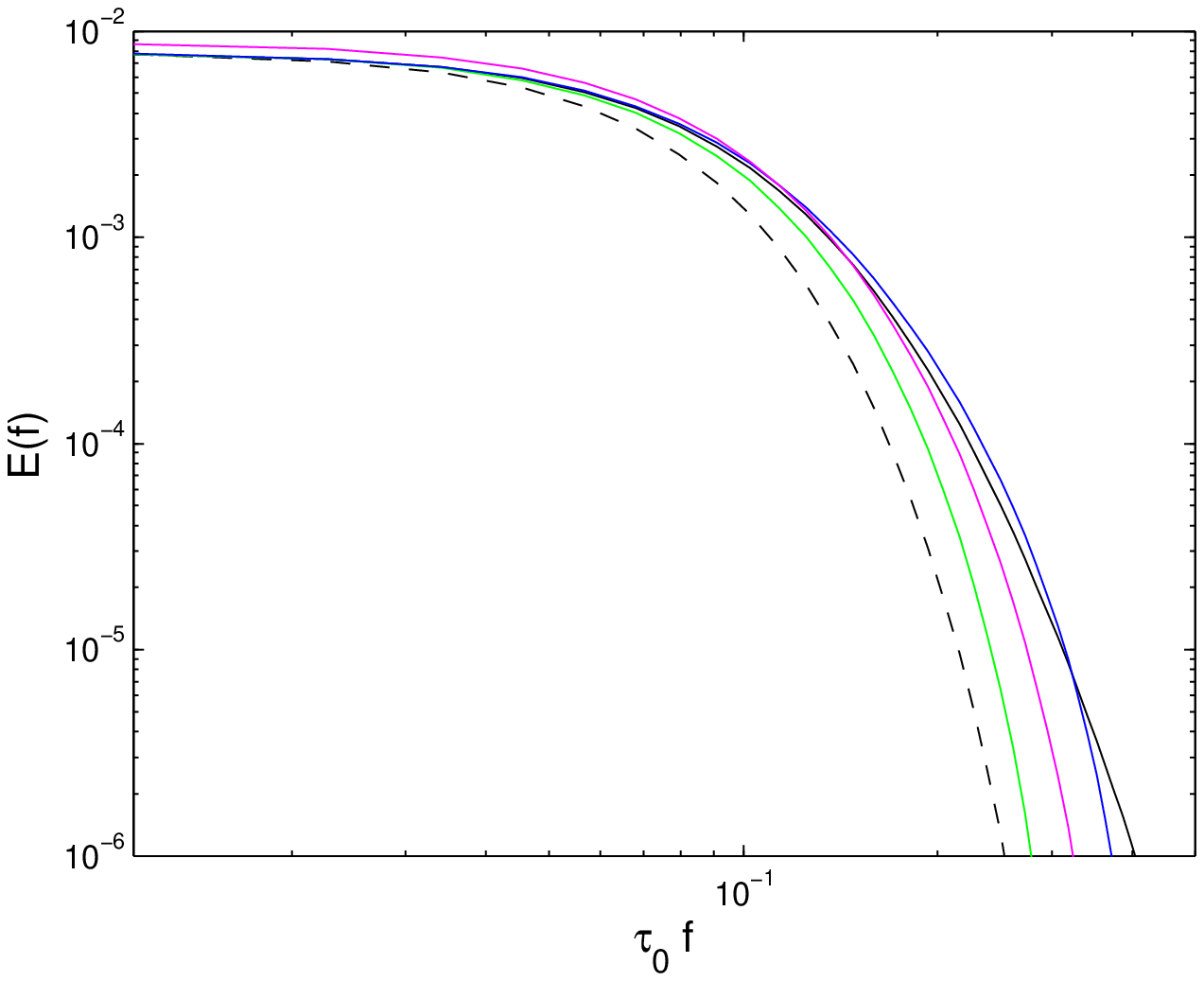}
\caption{Surface elevation energy density spectra, $E(f)$, as a function of the dimensionless frequency $\tau_0 f$, for a solitary wave, $\eps_0=a_0/h_0=0.4$, $\delta_m=0$. black line: numerical solution of the Euler equations; dashed black line: hydrostatic reconstruction $\zeta_{\rm H}$, equation \eqref{rec_Hyb}; green line: $\zeta_{\rm SL}$, equation \eqref{rec_linSWb}; blue  line: $\zeta_{\rm SNL}$, equation \eqref{rec_quadSWb}; magenta  line:  heuristic formula, $\zeta_{HE}$, equation \eqref{rec_heurb}. $\tau_0=h_0/c_p$, where $c_p$ is the solitary wave celerity.} 
\label{fig_S4}
\end{center}
\end{figure}

Figure \ref{fig_S1} presents a comparison between the Euler solitary wave elevation for $\eps_0=0.1$, and the shallow water reconstructions  \eqref{rec_Hyb}, \eqref{rec_linSWb} and \eqref{rec_quadSWb}. We can note that the hydrostatic reconstruction  \eqref{rec_Hyb} significantly underestimates the maximum wave elevation.  For this weakly-nonlinear long wave the linear shallow water reconstruction \eqref{rec_linSWb} gives good results although the maximum elevation is slightly underestimated. By taking into account the nonlinear effects  the reconstruction \eqref{rec_quadSWb} gives an excellent agreement with the Euler solution. For a high amplitude solitary wave ($\eps_0=0.4$) the linear reconstruction fails to reproduce the elevation profile (see figure \ref{fig_S2}). By contrast the nonlinear shallow water reconstruction \eqref{rec_quadSWb} agrees well with the Euler solution. The fully non-dispersive reconstruction \eqref{rec_quad0b} is not plotted because it gives similar results to the nonlinear shallow water reconstruction \eqref{rec_quadSWb}. Figure \ref{fig_S3} shows that the nonlinear heuristic  reconstruction \eqref{rec_heurb} proposed by \cite{oliveras2012} gives a good estimate of the maximum wave elevation. However, contrary to the nonlinear shallow water reconstruction \eqref{rec_quadSWb} the heuristic formula leads to wave solutions which are significantly less peaked than the Euler solutions. 

\begin{table}
\noindent\includegraphics[width=12cm]{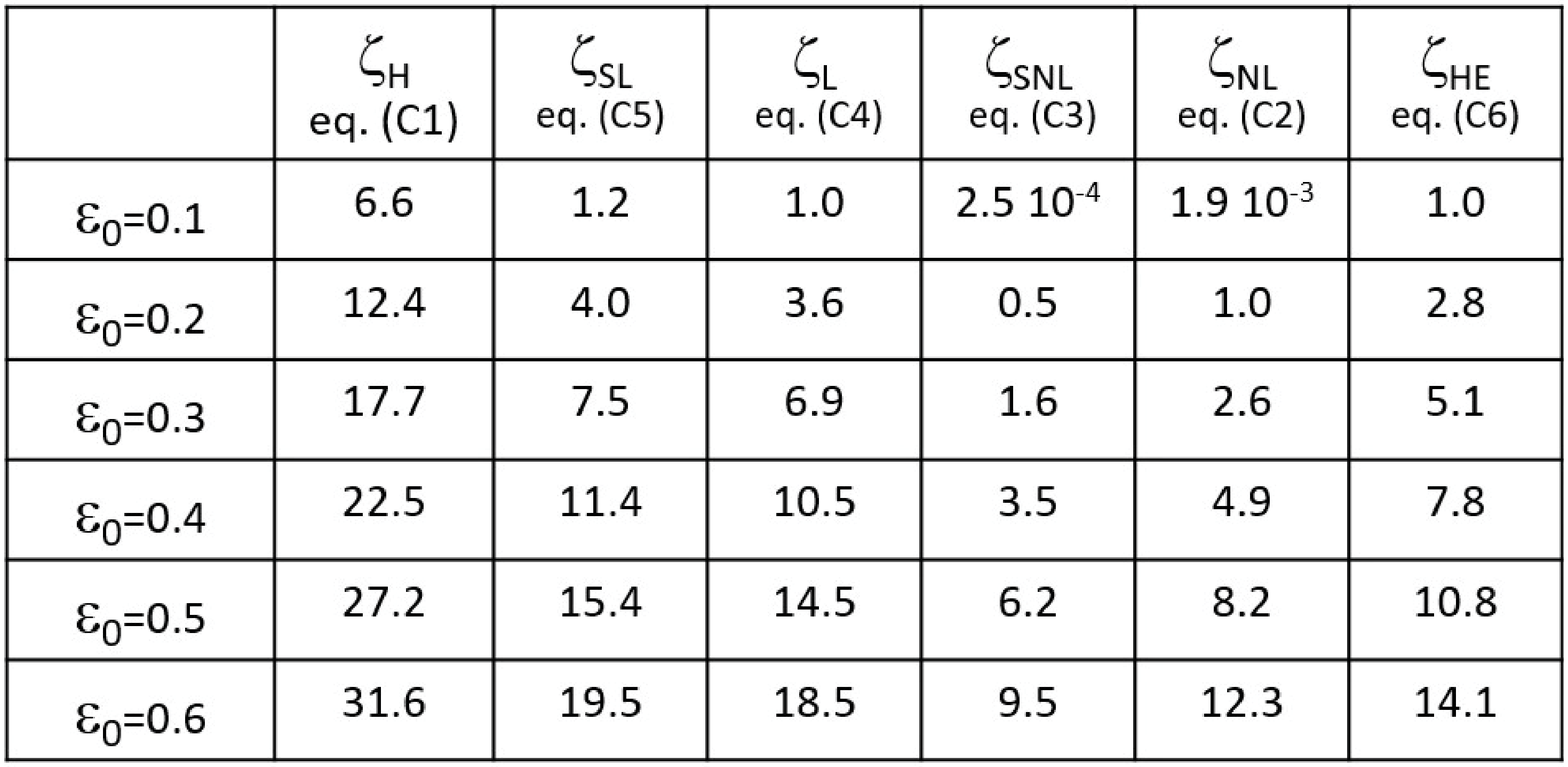}
\caption{ Normalized root mean square errors, RMSE, of the reconstruction formulas applied to solitary waves  of different intensities.}
\label{table1}
\end{table}

To quantify this observation we have computed the normalized root mean square error, $RMSE$, of the reconstruction formulas (see table \ref{table1}), applied to solitary waves of different intensities ($\eps_0$ ranging from 0.1 to 0.6). 
This error is defined by
$$
RMSE=\left(\frac{\langle (\zeta_R- \zeta )^2 \rangle}{\langle \zeta-\langle \zeta \rangle \rangle^2}\right)^{1/2} ,
$$
where $\zeta_R$ is a reconstructed wave elevation and $\langle . \rangle=\frac{1}{t_2-t_1}\int_{t_1}^{t_2} (.) dt$, with  $t_1=-5\tau_1$, $t_2=5\tau_1$, $\tau_1=\tau_0 / \sqrt{\eps_0}$ is the characteristic solitary wave duration  and $\tau_0=h_0/c_p$. 
 Table \ref{table1} shows, as a matter of course, that the $RMSE$ is an increasing function of $\eps_0$. We can see that the nonlinear formulas \eqref{rec_quadSWb}, \eqref{rec_quad0b}, \eqref{rec_heurb} significantly improve the reconstructed solution in comparison with linear reconstructions.  However, our nonlinear formulas \eqref{rec_quadSWb} and \eqref{rec_quad0b} give better results than  the heuristic formula \eqref{rec_heurb}, especially for low $\eps_0$.

A comparison of the surface elevation energy density spectra, $E(f)$, computed from the reconstruction formulas for $\eps_0=0.4$, is presented in figure \ref{fig_S4}. The shallow water linear formula \eqref{rec_linSWb}  properly reconstructs $E(f)$ at low frequencies but underestimates it at high frequencies. The heuristic reconstruction  \eqref{rec_heurb} overestimates  $E(f)$ at low frequencies (and thus also overestimates the mean elevation) and underestimates it at high frequencies. By contrast, our new nonlinear formula \eqref{rec_quadSWb} gives good results over the whole range of frequencies.

%
%
\subsection{Bichromatic wave}\label{sectbichro}

The ability of our formulas to reconstruct non-permanent form waves is assessed with respect to a bichromatic wave field propagating over a gently sloping movable bed.  This laboratory dataset is presented in \cite{michallet2017}. The two frequencies composing the wave-board motion were $f_1=0.5515$ Hz and $f_2=0.6250$ Hz, and the amplitude of the two wave components were identical with a value of 0.03 m.  The still water depth at the wave maker was $0.566$ m. Wave elevation and bottom pressure were synchronously measured in the shoaling zone at 18.5 m from the wave maker, corresponding to a still water depth of $h_0=0.326$ m ($\mu = 0.53$). The pressure sensor was located at $\delta_m=0.5$ cm above the bed.

\begin{figure}
\begin{center}
\includegraphics[width=14cm]{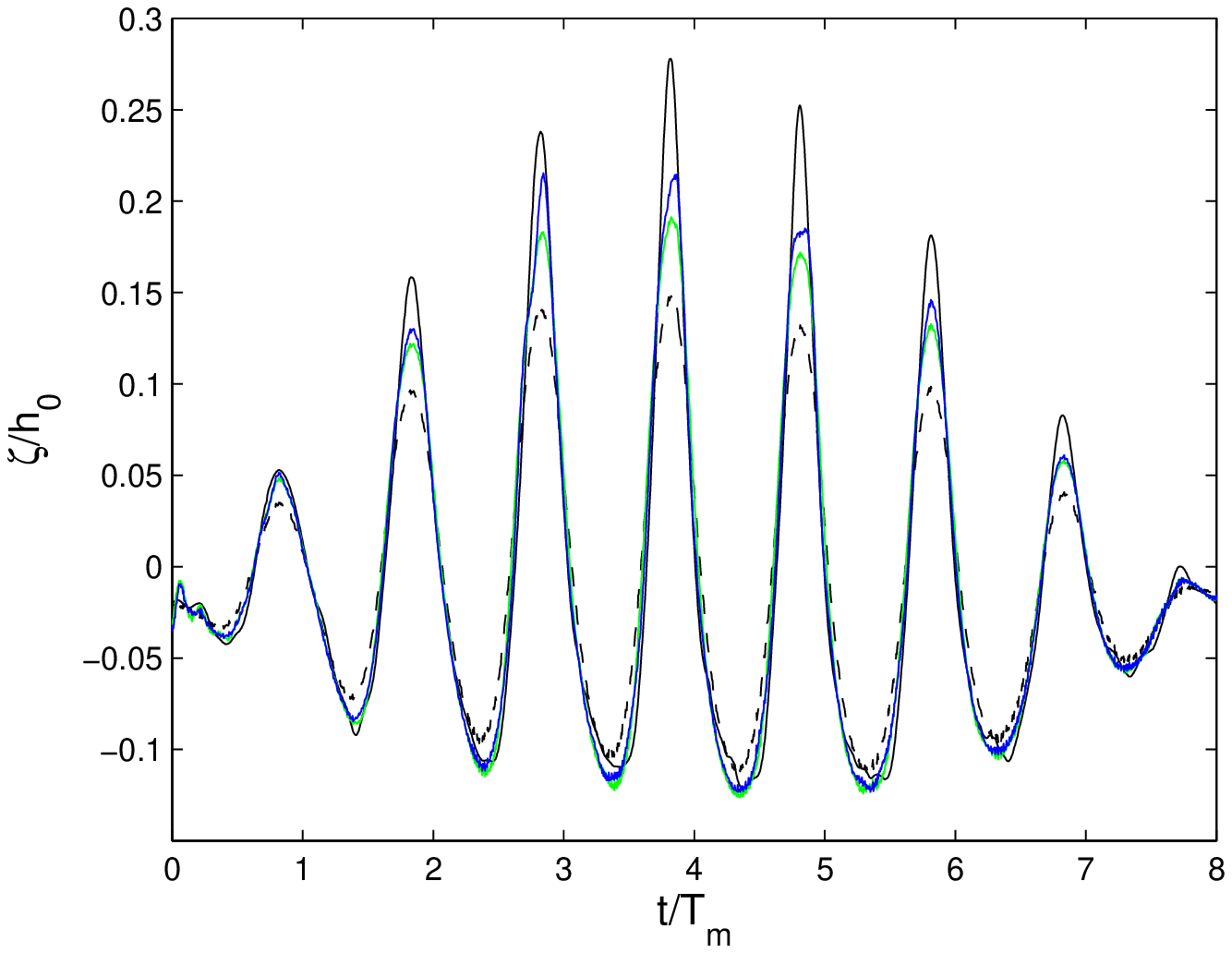}
\caption{Surface elevation reconstruction of bichromatic waves, $f_1=0.5515$ Hz, $f_2=0.6250$ Hz ($T_m=\left( \frac{f_1+f_2}{2}\right)^{-1}$), $h_0=0.326$ m and $\delta_m=0.5$ cm. black line: direct measurement of $\zeta$ ; dashed black line:  hydrostatic reconstruction $\zeta_{\rm H}$, equation \eqref{rec_Hyb}; green line: $\zeta_{\rm SL}$, equation \eqref{rec_linSWb2}; blue  line:  $\zeta_{\rm SNL}$, equation \eqref{rec_quadSWb}.} 
\label{fig_M1a}
\end{center}
\end{figure}

\begin{figure}
\begin{center}
\includegraphics[width=14cm]{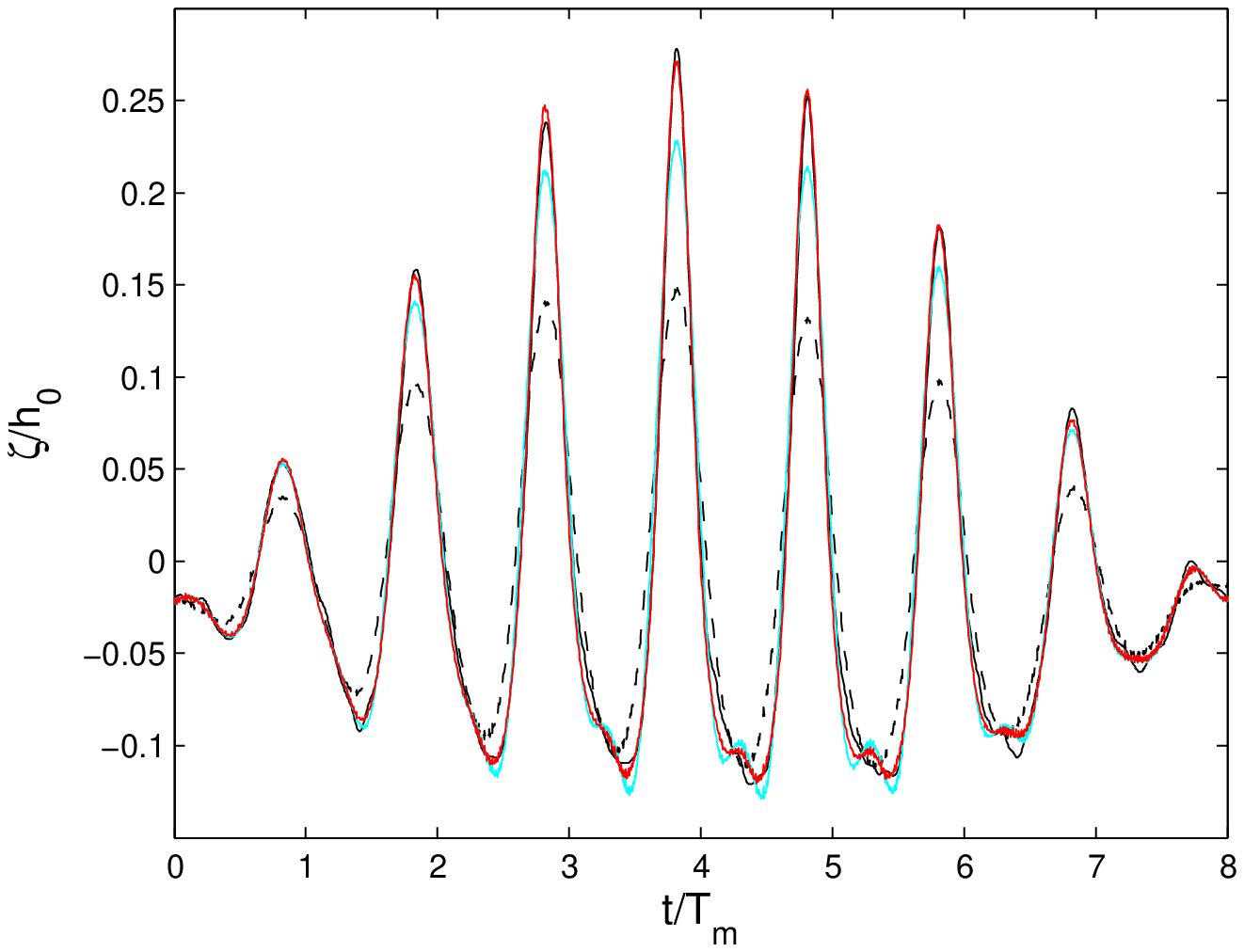}
\caption{Surface elevation reconstruction of bichromatic waves, $f_1=0.5515$ Hz, $f_2=0.6250$ Hz ($T_m=\left( \frac{f_1+f_2}{2}\right)^{-1}$), $h_0=0.326$ m and $\delta_m=0.5$ cm. cut-off frequency $f_c=1.5$ Hz. black line: direct measurement of $\zeta$ ; dashed black line:  hydrostatic reconstruction $\zeta_{\rm H}$, equation \eqref{rec_Hyb}; cyan line: $\zeta_{\rm L}$, equation \eqref{rec_lin1_meas}; red  line:  $\zeta_{\rm NL}$, equation \eqref{rec_quad0b}.} 
\label{fig_M1c}
\end{center}
\end{figure}

Figure \ref{fig_M1a} presents a comparison between shallow water reconstructions and  direct elevation measurements. The  linear shallow water reconstruction \eqref{rec_linSWb2} strongly underestimates the measured elevation especially for the highest waves of the wave group. The nonlinear reconstruction  \eqref{rec_quadSWb} improves the results but still  significantly underestimates the maximum wave elevation. These discrepancies do not question our nonlinear reconstruction approach but show the limitation of shallow water methods for describing a wave field with such a high value of the shallowness parameter ($\mu = 0.53$). To properly reconstruct the surface elevation of such dispersive waves a fully dispersive approach is required. The frequency cut-off $f_c$, introduced in section \ref{WNlinearwave}, is set to a value of 1.5 Hz. The bichromatic wave field is much better described by the fully dispersive linear reconstruction \eqref{rec_lin1_meas} (see figure \ref{fig_M1c}, cyan line) than by the shallow water linear reconstruction \eqref{rec_linSWb2} (figure \ref{fig_M1a}, green line). In figure \ref{fig_M1c} we can see that the fully dispersive linear formula  \eqref{rec_lin1_meas} gives excellent results for the lowest  waves of the wave group but significantly underestimates the elevation at the crest of the highest waves. By contrast our nonlinear formula \eqref{rec_quad0b} gives excellent results even for the highest waves. Figure \ref{fig_M1d} presents a zoom of the preceding figure on the highest wave of the wave group. In this figure we show that the linear reconstruction fails to reproduce the maximum elevation and above all the skewed shape of this nonlinear wave. By contrast the nonlinear formula \eqref{rec_quad0b} gives excellent results for both the maximum elevation and the horizontal asymmetry (or wave skewness). Figure \ref{fig_M1db} shows that the heuristic formula \eqref{rec_heurb2} gives good results for the maximum wave elevation but fails to describe the wave skewness. Indeed this wave property is largely controlled in our nonlinear reconstruction \eqref{rec_quad0b} by the term $\sqrt{\mu}\sigma\big(\dt \zeta_{\rm L}\big)^2$ which is missing in the heuristic formula (see equation\eqref{rec_heurc}). In order to quantify the horizontal asymmetry of the highest wave we have computed the skewness parameter:
$$
S_k=\frac{\langle (\zeta-\langle \zeta \rangle )^3 \rangle}{\langle (\zeta-\langle \zeta \rangle)^2  \rangle^{3/2}} ,
$$
where $\langle . \rangle=\frac{1}{t_2-t_1}\int_{t_2}^{t_1} (.) dt$, $t_1$ and $t_2$ being the times of passage of the wave troughs surrounding the highest crest. This parameter is equal to zero for a sinusoidal wave. Table \ref{table2} shows that the linear reconstruction strongly underestimates, by $25 \%  $, the wave skewness. The nonlinear reconstructions significantly improve the results with an error of $11 \%$ for the heuristic formula and only $3 \%$ for the formula \eqref{rec_quad0b}.

Figures \ref{fig_M1e} and \ref{fig_M1f} show a comparison between the measured surface elevation energy density spectrum and the spectra obtained from reconstruction formulas. We can see in figure \ref{fig_M1e} that  the linear \eqref{rec_lin1_meas} and nonlinear \eqref{rec_quad0b} reconstructions properly describe the elevation energy around the first (i.e. fundamental) and second harmonics.  For $f > f_c$, the nonlinear reconstruction \eqref{rec_quad0b} is able to accurately fill the energy around the third and fourth harmonics.  By contrast, the nonlinear heuristic formula \eqref{rec_heurb2} gives a good prediction for the first  and second harmonics but significantly underestimates the energy for the higher harmonics.

\begin{figure}
\begin{center}
\includegraphics[width=14cm]{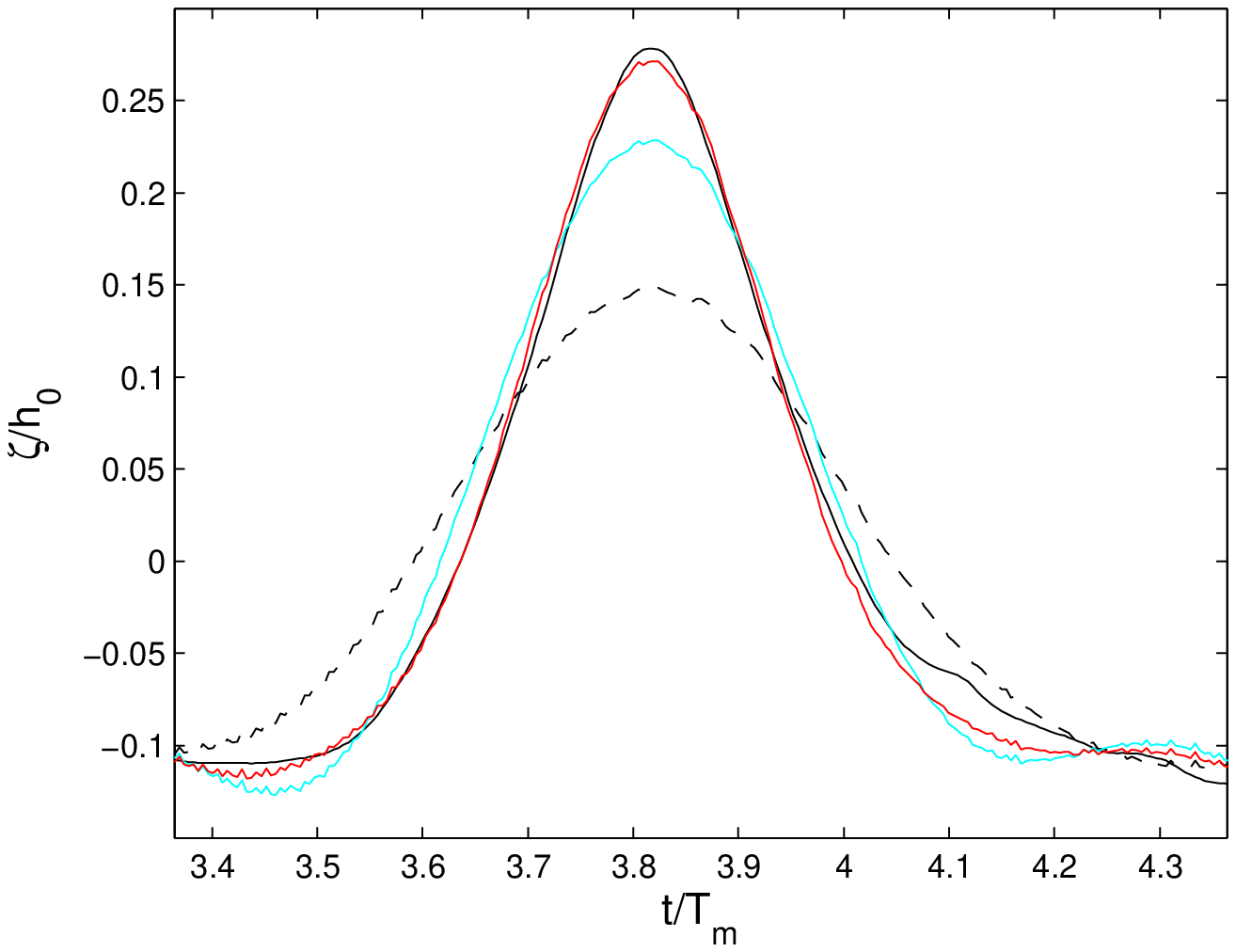}
\caption{Surface elevation reconstruction of bichromatic waves, $f_1=0.5515$ Hz, $f_2=0.6250$ Hz ($T_m=\left( \frac{f_1+f_2}{2}\right)^{-1}$), $h_0=0.326$ m and $\delta_m=0.5$ cm; zoom on the highest  wave of the wave group. cut-off frequency $f_c=1.5$ Hz. black line: direct measurement of $\zeta$ ; dashed black line:  hydrostatic reconstruction $\zeta_{\rm H}$, equation \eqref{rec_Hyb}; cyan line: $\zeta_{\rm L}$, equation \eqref{rec_lin1_meas}; red  line:  $\zeta_{\rm NL}$, equation \eqref{rec_quad0b}.} 
\label{fig_M1d}
\end{center}
\end{figure}

\begin{figure}
\begin{center}
\includegraphics[width=14cm]{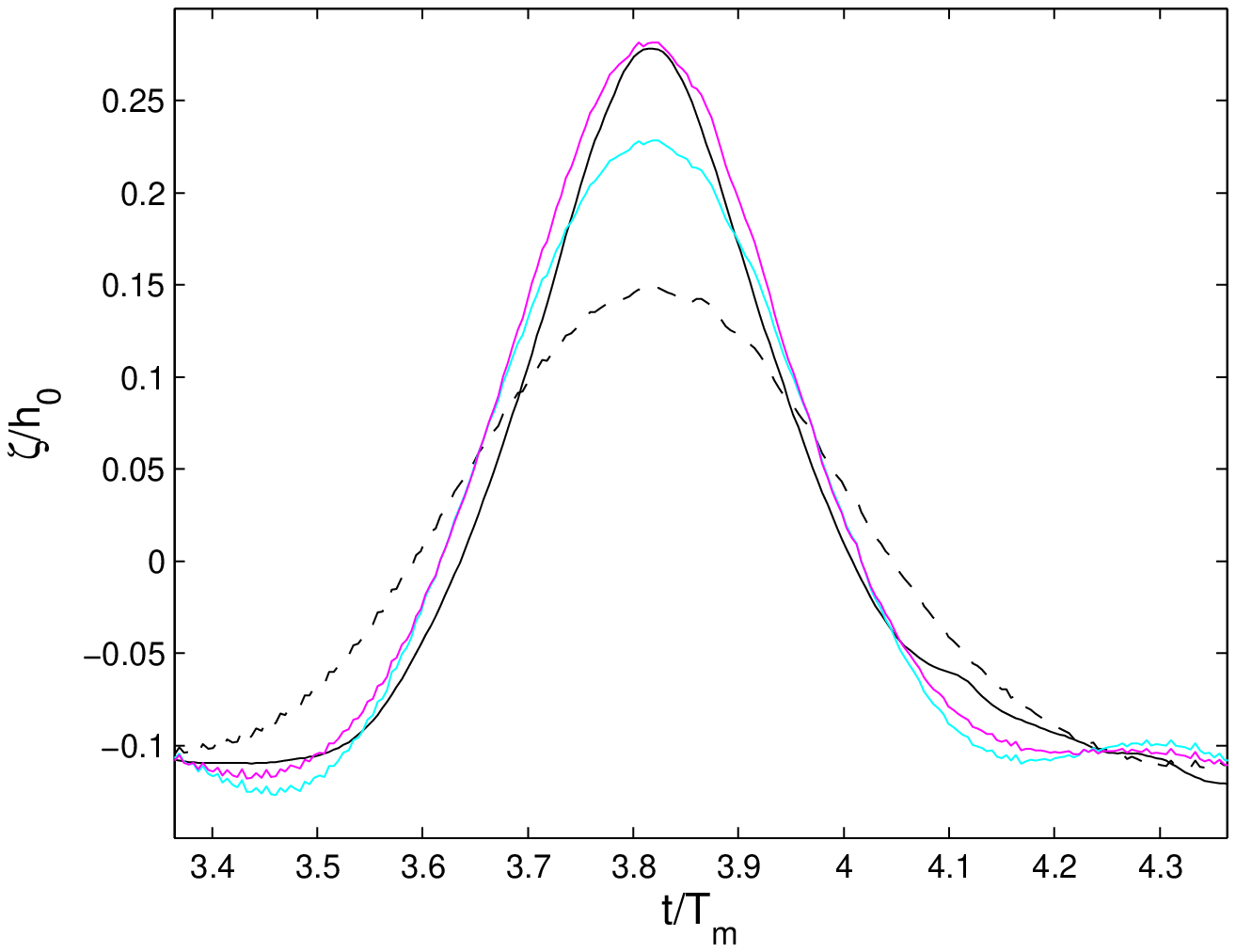}
\caption{Surface elevation reconstruction of bichromatic waves, $f_1=0.5515$ Hz, $f_2=0.6250$ Hz ($T_m=\left( \frac{f_1+f_2}{2}\right)^{-1}$), $h_0=0.326$ m and $\delta_m=0.5$ cm; zoom on the highest  wave of the wave group. cut-off frequency $f_c=1.5$ Hz. black line: direct measurement of $\zeta$ ; dashed black line:  hydrostatic reconstruction $\zeta_{\rm H}$, equation \eqref{rec_Hyb}; cyan line: $\zeta_{\rm L}$, equation \eqref{rec_lin1_meas}; magenta  line:  heuristic formula, $\zeta_{HE}$, equation \eqref{rec_heurb2}.} 
\label{fig_M1db}
\end{center}
\end{figure}

\begin{figure}
\begin{center}
\includegraphics[width=14cm]{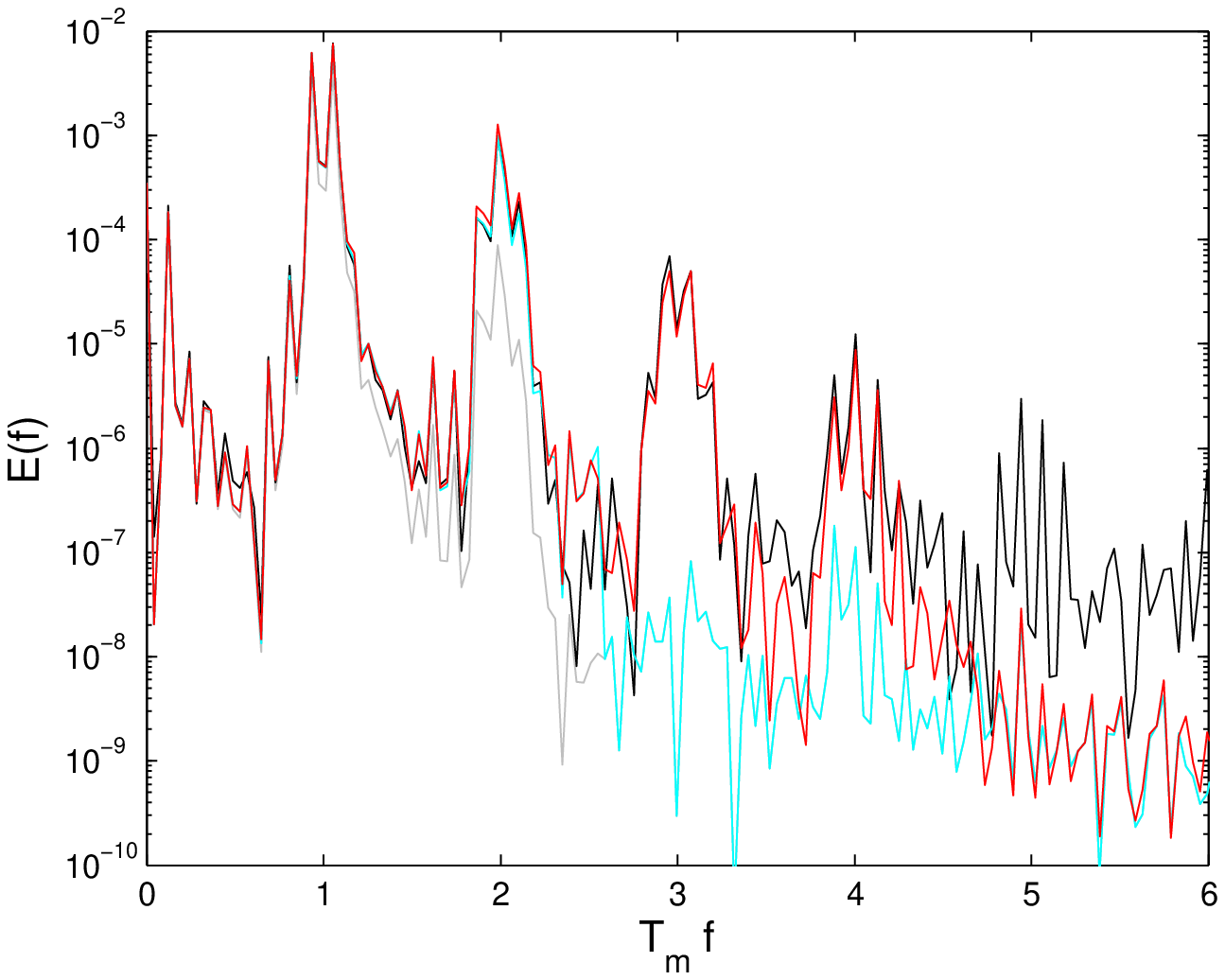}
\caption{Surface elevation energy density spectra, $E(f)$, as a function of the dimensionless frequency $T_m f$, for bichromatic waves, $f_1=0.5515$ Hz, $f_2=0.6250$ Hz ($T_m=\left( \frac{f_1+f_2}{2}\right)^{-1}$), $h_0=0.326$ m and $\delta_m=0.5$ cm. cut-off frequency $f_c=1.5$ Hz. black line: direct measurement of $\zeta$ ; grey line:  hydrostatic reconstruction $\zeta_{\rm H}$, equation \eqref{rec_Hyb}; cyan line: $\zeta_{\rm L}$, equation \eqref{rec_lin1_meas}; red  line:  $\zeta_{\rm NL}$, equation \eqref{rec_quad0b}.} 
\label{fig_M1e}
\end{center}
\end{figure}

\begin{figure}
\begin{center}
\includegraphics[width=14cm]{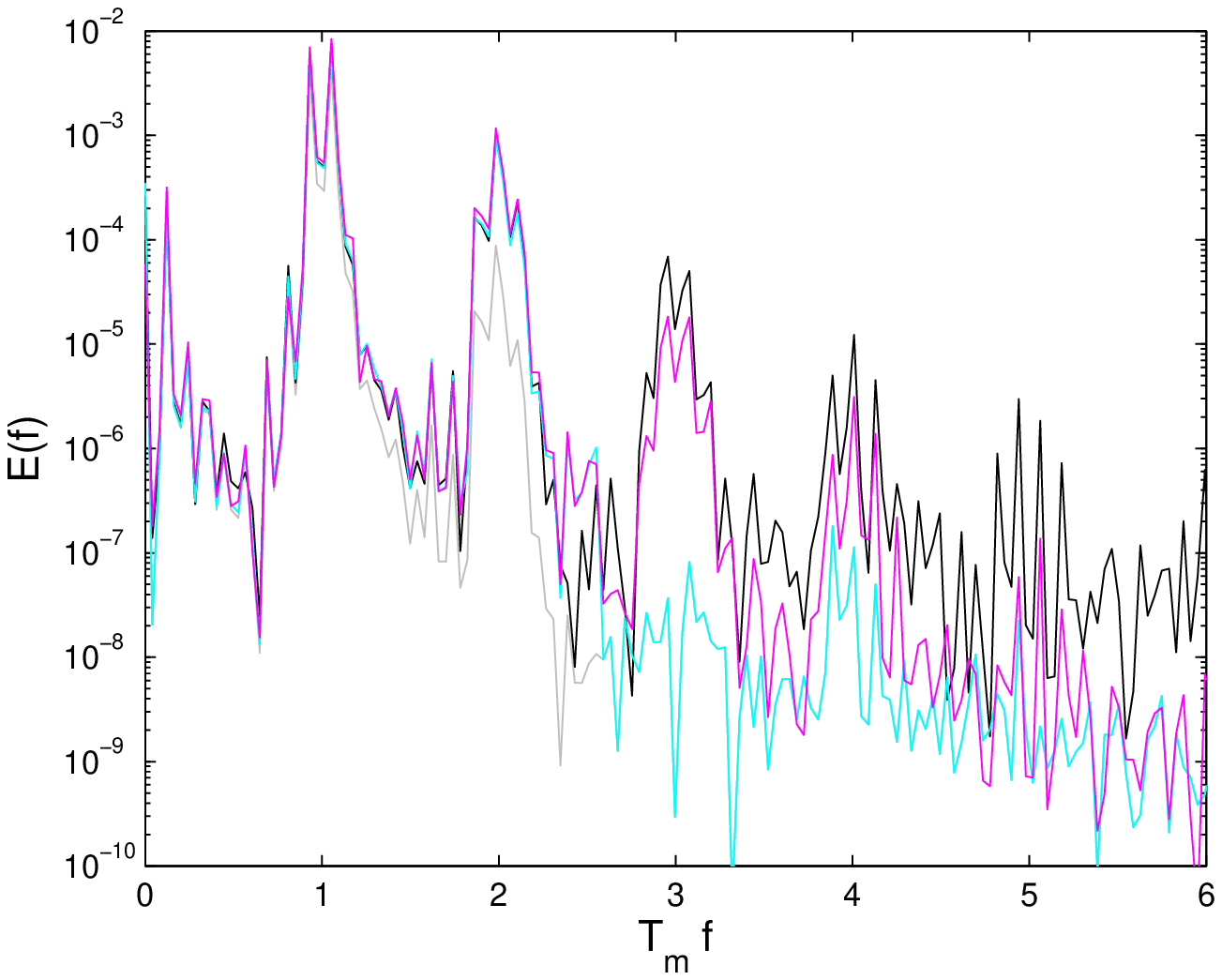}
\caption{ Surface elevation energy density spectra, $E(f)$, as a function of the dimensionless frequency $T_m f$, for bichromatic waves, $f_1=0.5515$ Hz, $f_2=0.6250$ Hz ($T_m=\left( \frac{f_1+f_2}{2}\right)^{-1}$), $h_0=0.326$ m and $\delta_m=0.5$ cm. cut-off frequency $f_c=1.5$ Hz. black line: direct measurement of $\zeta$ ; grey line:  hydrostatic reconstruction $\zeta_{\rm H}$, equation \eqref{rec_Hyb}; cyan line: $\zeta_{\rm L}$, equation \eqref{rec_lin1_meas}; magenta  line:  heuristic formula, $\zeta_{HE}$, equation \eqref{rec_heurb2}} 
\label{fig_M1f}
\end{center}
\end{figure}

\begin{table}
\small
\begin{tabular}{|l|p{2.5cm}|p{2.5cm}|p{2.5cm}|p{2.5cm}|p{2.7cm}|r|}
\hline
 &\ \  measurements  &\ \ \ $\zeta_{\rm L}$, eq. \eqref{rec_lin1_meas} &\ \ $\zeta_{HE}$, eq. \eqref{rec_heurb2}  & \ \  $\zeta_{\rm NL}$, eq.  \eqref{rec_quad0b}  \\
\hline
\ \  $S_k$  & \ \ \ \ \ \ \ 0.93  &\ \ \ \ \ \ \ 0.70 & \ \ \ \ \ \ \ 0.83 & \ \ \ \ \ \ \  0.96  \\
 \hline
   $S_k$ error& $\ $ &\ \ \ \ \ \ \ $25\%$ &\ \ \ \ \ \ \ $11\%$ &\ \ \ \ \ \ \ $3\%$  \\
\hline
\end{tabular}
\caption{ Sea surface skewness. Comparison between between reconstructed elevation and direct elevation measurements.}
\label{table2}
\end{table}

%
%
%
\section{Conclusion}
We have derived a weakly-nonlinear fully dispersive reconstruction formula, equation \eqref{rec_quad0}, which allows the elevation reconstruction of ocean waves in intermediate and shallow waters. These formulas involve a Fourier multiplier in space which requires the knowledge of the pressure field, $P(t,X)$, over the whole horizontal space. For most ocean applications, the pressure is only known at one measurement point. To overcome this limitation, we have shown in section \ref{practical} how to replace the Fourier multiplier in space by a Fourier multiplier in time. Two distinct cases have been considered: nonlinear permanent form waves (i.e. traveling waves) and irregular weakly nonlinear waves. The former corresponds to an academic case, which is useful to validate the nonlinear reconstruction formula, and the later case is essential in the context of wind-wave applications. To test the ability of our approach we have compared reconstructed surface elevation fields to numerical Euler solutions and wave-tank experiments. We have shown that our nonlinear method provides much better  results of the surface elevation reconstruction  compared to the transfer function approach commonly used in coastal applications. In particular, our method accurately reproduces the maximum elevation and the skewed shape of nonlinear wave fields. These properties are essential for applications such as those on extreme waves and wave-induced sediment transport. Moreover our reconstruction formulas are simple and easy to use for operational and real time coastal applications.

Our nonlinear method mainly applies to wind-generated waves propagating in coastal environments where the effects of background currents can be neglected. A first attempt to generalize this approach in presence of a vertically- and horizontally-uniform current is presented in appendix \ref{app_current}. However, further studies are required if we are to reconstruct nonlinear waves in presence of horizontally-variable currents (e.g. tsunami and tidal bores) or vertically-variable currents where the long-wave model with vorticity derived by \cite{castro} could be used.

%
%
%
\section*{Acknowledgments}

 Both authors have been partially funded by the ANR-
13-BS01-0009-01 BOND. David Lannes also acknowledges support from the ANR-13-BS01-0003-01 DYFICOLTI and Philippe Bonneton from ANR-14-ASTR-0019 COASTVAR.

We want to warmly thank Herv\'e Michallet (LEGI, France) for kindly providing us with his bichromatic  wave dataset.

%
%
\appendix


\section{Asymptotic expansion of the velocity potential $\Phi$}\label{appasPhi}
\label{appendix2} 
We know from \eqref{Laplace_ND} that the dimensionless velocity potential $\Phi$ solves the boundary value problem
$$
\begin{cases}
\mu \Delta\Phi+\dz^2\Phi=0 &\mbox{in }\Omega,\\
\Phi_{\vert_{z=\eps\zeta}}=\psi, &\dz \Phi_{\vert_{z=-1}}=0.
\end{cases}
$$
Our goal here is to prove the following two approximations
$$
\Phi=\Phi^0+O(\sigma),
\quad \mbox{ and }\quad \Phi=\Phi^0+\eps\Phi^1+O(\sigma^2),
$$
with
\begin{equation}\label{formulaPhi}
\Phi^0=\frac{\cosh(\sqrt{\mu}(z+1)\abs{D})}{\cosh(\sqrt{\mu}\abs{D})}\psi
\quad\mbox{ and }\quad
\Phi^1=-\frac{\cosh(\sqrt{\mu}(z+1)\abs{D})}{\cosh(\sqrt{\mu}\abs{D})} (\zeta G_0\psi).
\end{equation}
In order to compare the exact solution to these two approximations, it is convenient to work in the flat strip ${\mathcal S}=\R^d\times (-1,0)$ instead of the fluid domain $\Omega$; this is possible through the change of variable
$$
\forall (X,z)\in {\mathcal S},\qquad \phi(X,z)=\Phi(X, (z+1)\eps\zeta(X) +z).
$$
From the equation solved by $\Phi$, one deduces that $\phi$ must solve the boundary value problem
$$
\begin{cases}
\nam\cdot P(\zeta)\nam\phi=0,\quad\mbox{ in }{\mathcal S},\\
\phi_{\vert_{z=0}}=\psi,\qquad \dz \phi_{\vert_{z=-h_0}}=0,
\end{cases}
$$
with 
$$
\nam=\left(\begin{array}{c}\sqrt{\mu}\nabla \\ \dz \end{array}\right)
\quad\mbox{ and }\quad P(\zeta)=\left(
\begin{array}{cc}
(1+\eps\zeta)I_2 & -\sigma (z+1)\nabla \zeta\\
-\sigma(z+1)(\nabla\zeta)^T & \frac{1+\sigma^2(z+1)^2 \vert \nabla\zeta\vert^2}{1+\eps\zeta}
\end{array}\right).
$$
It is known that $\phi$ depends analytically on $\zeta$; therefore, one can write
$$
\phi=\phi^0+\eps\phi^1+\phi^{\geq 2} \quad\mbox{ with }\quad \phi^{\geq 2}:=\sum_{j=2}^\infty \eps^j \phi^j,
$$
and where $\phi^j$ is $j$-linear in $\zeta$ and therefore $j+1$ linear in $(\zeta,\psi)$. Expanding the matrix $P(\zeta)$ in terms of $\zeta$, one also gets
$$
P(\zeta)=P^0+\eps P^1+P^{\geq 2},
$$
with $P^0=\mbox{Id}$ and
$$
P^1=\left(
\begin{array}{cc}
\zeta & -(z+1)\sqrt{\mu}\nabla \zeta\\
-(z+1)\sqrt{\mu}(\nabla\zeta)^T & -\zeta
\end{array}\right)
\qquad
P^{\geq 2}=\left(
\begin{array}{cc}
0 & 0\\
0 & \frac{\eps^2\zeta^2+\sigma^2(z+1)^2 \vert \nabla\zeta\vert^2}{1+\eps\zeta}
\end{array}\right).
$$
\subsection{Linear approximation}
The linear (with respect to $(\zeta,\psi)$) approximation $\phi^0$ of the velocity potential is then found by solving 
$$
\begin{cases}
\nam\cdot P^0\nam\phi^0=0,\quad\mbox{ in }{\mathcal S},\\
\phi^0_{\vert_{z=0}}=\psi,\qquad \dz \phi^0_{\vert_{z=-1}}=0,
\end{cases}
$$
so that one  finds 
\begin{equation}\label{linapp}
\phi^0(X,z)=\frac{\cosh(\sqrt{\mu}(z+1)\abs{D})}{\cosh(\sqrt{\mu}\abs{D})}\psi.
\end{equation}
Moreover, since the approximation error ${\bf e}^0:=\phi-\phi^0$ solves
$$
\begin{cases}
\nam\cdot P^0\nam{\bf e}^0=- \nam \cdot (\eps P^1+P^{\geq 2})\nam \phi,\quad\mbox{ in }{\mathcal S},\\
{\bf e}^0_{\vert_{z=0}}=0,\qquad \dz {\bf e}^0_{\vert_{z=-1}}=0,
\end{cases}
$$
and since $\eps P^1+P^{\geq 2}=O(\eps)$ and $\nam \phi=O(\sqrt{\mu})$ (see for instance Proposition 2.36 in \cite{lannes2013}), one deduces easily that $\nam {\bf e}^0=O(\sigma)$ and therefore
$$
\phi-\phi^0=O(\sigma).
$$
\subsection{Quadratic approximation}
The system governing the quadratic term $\phi^1$ is then found by isolating the terms of order $1$ in $\zeta$,
$$
\begin{cases}
\nam \cdot P^0\nam\phi^1=- \nam\cdot P^1\nam\phi^0,\quad\mbox{ in }{\mathcal S},\\
\phi^1_{\vert_{z=0}}=0,\qquad \dz \phi^1_{\vert_{z=-h_0}}=0.
\end{cases}
$$
Solving directly this system leads to very complicated expressions; however, many simplifications arise from the observation that
$$
- \nam\cdot P^1\nam\phi^0=\nam\cdot P^0\nam \tilde \phi^1,\quad\mbox{ with }\quad \tilde\phi^1=(z+1)\zeta \dz \phi^0.
$$
Indeed, the difference ${\mathtt r}:=\phi^1-\tilde \phi^1$ solves
$$
\begin{cases}
\nam\cdot P^0\nam{\mathtt r}=0,\quad\mbox{ in }{\mathcal S},\\
\phi_{\vert_{z=0}}=-\eps\zeta \dz \phi^0_{\vert_{z=0}},\qquad \dz \phi_{\vert_{z=-h_0}}=0.
\end{cases}
$$
Proceeding as for the linear approximation, with $\psi$ replaced by $-\eps\zeta \dz \phi^0_{\vert_{z=0}}$, this leads to
\begin{align*}
{\mathtt r}&=-\frac{\cosh(\sqrt{\mu}(z+1)\abs{D})}{\cosh({\sqrt{\mu}}\abs{D})} \big(\zeta \dz \phi^0_{\vert_{z=0}}\big)\\
&=-\frac{\cosh(\sqrt{\mu}(z+1)\abs{D})}{\cosh(\sqrt{\mu}\abs{D})} (\zeta G_0\psi),
\end{align*}
where we used \eqref{linapp} to derive the second identity. To summarize, we have thus proved that
\begin{equation}\label{intermbil}
\phi^1=(z+1)\zeta \dz \phi^0-\frac{\cosh(\sqrt{\mu}(z+1)\abs{D})}{\cosh(\sqrt{\mu}\abs{D})} (\zeta G_0\psi).
\end{equation}
Moreover, since the approximation error ${\bf e}^1:=\phi-\phi^0-\eps \phi^1$ solves
$$
\begin{cases}
\nam\cdot P^0\nam{\bf e}^1=- \nam \cdot P^{\geq 2}\nam \phi-\eps^2 \nam \cdot P^1 \nam \phi^1 ,\quad\mbox{ in }{\mathcal S},\\
{\bf e}^1_{\vert_{z=0}}=0,\qquad \dz {\bf e}^1_{\vert_{z=-1}}=0,
\end{cases}
$$
and since $\nam \cdot P^{\geq 2}\nam \phi=O(\sigma^2)$ and $\nam \phi_1=O({\mu})$, one deduces easily that $\nam {\bf e}^1=O(\sigma^2)$ and therefore
$$
\phi-\phi^0-\eps\phi^1=O(\sigma^2).
$$
\subsection{Proof of the approximations in the physical domain}
Let us now go back to the velocity potential in the physical domain $\Omega$ by undoing our change of variables,
$$
\forall (X,z)\in \Omega,\qquad \Phi(X,z)=\phi \Big(X,\frac{z-\eps\zeta}{1+\eps\zeta}\Big).
$$
By simple Taylor expansions, we get that
$$
\Phi(X,z)=\Phi^0(X,z)+\eps \Phi^1(X,z)+O(\sigma^2),
$$
with $\Phi^0$ given by \eqref{formulaPhi} and $\Phi^1$ given by
$$
\Phi^1=-(z+1)\zeta\dz \Phi^0+\phi^1.
$$
Using  \eqref{intermbil}, we finally obtain that the bilinear term in our approximation of the velocity potential $\Phi$ is as given by \eqref{formulaPhi}.


\section{Pressure measured at an arbitrary depth }\label{appnotbot}

The hydrostatic,  linear and shallow water reconstruction formulas express the surface elevation in terms of the pressure $P_{\rm b}$ measured at the bottom $z=-1$ (in dimensionless variables). However, sensors are very often located above the bottom at some distance $\delta_m >0$.  We show here how these reconstruction formulas must be adapted in order to be expressed in terms of the pressure $P_{\rm m}$ measured at this point $z=z_{\rm m}=-1+\delta_m$.\\
Evaluating the Bernoulli equation \eqref{Bernoulli_ND} at  $z=z_{\rm m}$ rather than at the bottom, we are led to replace \eqref{eqbott_ND} by
$$
\dt \Phi_{\rm m}+\frac{1}{\eps}z_{\rm m}+\frac{\eps}{2}\vert \nabla\Phi_m\vert^2+\frac{\eps}{2\mu}\abs{\dz \Phi_{\vert_{z=z_{\rm m}}}}^2=-\frac{1}{\eps}(P_{\rm m}-P_{\rm atm}),
$$
where $\Phi_{\rm m}=\Phi{\vert_{z=z_{\rm m}}}$;
consequently, \eqref{formuleexacte_ND} is generalized into
\begin{align*}
\zeta=&\zeta_{\rm H}+\dt \Phi_{\rm m}-\dt \psi\\
&+\frac{\eps}{2}\big(\vert \nabla\Phi_{\rm m}\vert^2-\vert \nabla\psi\vert^2\big)+\frac{\eps}{2\mu}\abs{\dz \Phi_{\vert_{z=z_{\rm m}}}}^2+\frac{\eps}{2\mu} (1+ \eps^2\mu\vert \nabla\zeta\vert^2)(\dz \Phi_{\vert_{z=\eps\zeta}})^2 ,
\end{align*}
where the hydrostatic reconstruction $\zeta_{\rm H}$ is now given in terms of $P_{\rm m}$,
\begin{equation}\label{rec_hydro_ND_meas}
\zeta_{\rm H}=\frac{1}{\eps}(P_{\rm m}-P_{\rm atm}-1+\delta_m).
\end{equation}
Using \eqref{app0} and \eqref{eqsurf_ND}, this yields
\begin{align*}
\zeta&=\zeta_{\rm H}+\big(\frac{\cosh(\sqrt{\mu}(z_{\rm m}+1)\abs{D})}{\cosh(\sqrt{\mu}\abs{D})}-1\big)\dt \psi+O(\sigma)\\
&=\zeta_{\rm H}+\big(1-\frac{\cosh(\sqrt{\mu}(z_{\rm m}+1)\abs{D})}{\cosh(\sqrt{\mu}\abs{D})}\big)\zeta+O(\sigma),
\end{align*}
so that
\begin{equation}\label{rec_lin0_meas}
\zeta_{\rm L}=\frac{\cosh\big(\sqrt{\mu} \vert D\vert\big)}{\cosh(\sqrt{\mu}\delta_{\rm m}\vert D\vert)}\zeta_{\rm H}.
\end{equation}
Following the same path as in \S \ref{sectquadr}, the quadratic formula \eqref{rec_quad0} is generalized into
\begin{equation}\label{notbotNL}
\zeta_{\rm NL}=\zeta_{L}-\eps\mu \zeta_{\rm L}\dt \zeta_{\rm L}+\eps \mu \frac{\cosh(\sqrt{\mu}\vert D\vert)}{\cosh(\sqrt{\mu}\delta_m\vert D\vert)}\big( \frac{\sinh(\sqrt{\mu}\delta_m)}{\sinh(\sqrt{\mu}\vert D\vert)}\dt \zeta_{\rm L}\big)^2.
\end{equation}
In the shallow water regime ($\mu\ll 1$) these formulas can be simplified into the following generalization of \eqref{rec_linSW0},
\begin{equation}\label{rec_linSW0_meas}
\zeta_{\rm SL}=\zeta_{\rm H}-\frac{\mu}{2}(1-\delta_m^2)\Delta\zeta_{\rm H}
\end{equation}
for the linear formula, and of \eqref{rec_quadSW} for the quadratic formula,
\begin{equation}\label{rec_quadSWab}
\zeta_{\rm SNL}=\zeta_{\rm SL}-\eps\mu \dt \big(\zeta_{\rm SL} \dt \zeta_{\rm SL}\big)+\eps \mu \delta_m^2 (\dt \zeta_{\rm SL})^2.
\end{equation}


\section{Dimensional reconstruction formulas}\label{dimensional}
In variables with dimension, the hydrostatic reconstruction \eqref{rec_hydro_ND_meas} can be written 
\begin{equation}\label{rec_Hyb}
\zeta_{\rm H}=\frac{P_{\rm m}-P_{\rm atm}}{\rho g}+\delta_m-h_0,
\end{equation}


\subsection{Formulas involving Fourier multiplier in space}

In variables with dimension, the linear reconstructions \eqref{rec_lin0_meas} (fully dispersive) and \eqref{rec_linSW0_meas} (shallow water) can be written 

\begin{eqnarray*}
\zeta_{\rm L}&=&\frac{\cosh\big(h_0 \vert D\vert\big)}{\cosh\big(\delta_{\rm m}\vert D\vert\big)}\zeta_{\rm H} \\
\zeta_{\rm SL}&=&\zeta_{\rm H}-\frac{h_0^2}{2}\left(1-\left(\frac{\delta_m}{h_0}\right)^2\right)\Delta\zeta_{\rm H}.
\end{eqnarray*}

The fully dispersive \eqref{notbotNL} and shallow water \eqref{rec_quadSWab} nonlinear reconstructions write 
\begin{equation}\label{rec_quad0b}
\zeta_{\rm NL}=\zeta_{\rm L}-\frac{1}{g} \dt \big(\zeta_{\rm L} \dt \zeta_{\rm L}\big) +\frac{1}{g} \frac{\cosh(h_0\vert D\vert)}{\cosh(\delta_m\vert D\vert)}\big( \frac{\sinh(\delta_m)}{\sinh(h_0\vert D\vert)}\dt \zeta_{\rm L}\big)^2;
\end{equation}
and
\begin{equation}\label{rec_quadSWb}
\zeta_{\rm SNL}=\zeta_{\rm SL}-\frac{1}{g} \dt \big(\zeta_{\rm SL} \dt \zeta_{\rm SL}\big)+\frac{1}{g} \Big(\frac{\delta_m}{h_0}\Big)^2 (\dt \zeta_{\rm SL})^2.
\end{equation}

The nonlinear heuristic equation (see \cite{vasan2017} ) is given by

$$
\zeta_{HE}=\frac{\zeta_{\rm L}} {1- \frac{ D \sinh \big(h_0 D \big)}{\cosh\big(\delta_m D\big)}\zeta_{\rm H}}.
$$


\subsection{Formulas involving Fourier multiplier in time}

\subsubsection*{Permanent form wave}
The fully dispersive and shallow water linear reconstructions write respectively
\begin{equation}\label{rec_lin0_measb}
\zeta_{\rm L}=\frac{\cosh\big(\frac{h_0 \vert D_t\vert}{c_p}\big)}{\cosh\big(\frac{\delta_{\rm m}\vert D_t\vert}{c_p}\big)}\zeta_{\rm H},
\end{equation}
and
\begin{equation}\label{rec_linSWb}
\zeta_{\rm SL}=\zeta_{\rm H}-\frac{h_0^2}{2c_p^2}\left(1-\left(\frac{\delta_m}{h_0}\right)^2\right)\dt^2\zeta_{\rm H}.
\end{equation}

The heuristic equation writes
\begin{equation}\label{rec_heurb}
\zeta_{HE}=\frac{\zeta_{\rm L}} {1- \frac{ D_t \sinh \big(\frac{h_0 D_t}{c_p} \big)}{c_p\cosh\big(\frac{\delta_m D_t}{c_p}\big)}\zeta_{\rm H}}.
\end{equation}

\subsubsection*{Irregular waves}

Proceeding as in \S \ref{linearwave} by replacing the Fourier multiplier in space by a Fourier multiplier in time we obtain:
\begin{equation}\label{rec_lin1_meas}
\zeta_{\rm L}=\frac{\cosh\big(h_0 k(\vert D_t\vert)\big)}{\cosh\big(\delta_{\rm m}k(\vert D_t\vert))}\zeta_{\rm H},
\end{equation}
where $k(\omega)$ is given by the dispersion relation
$$
\omega^2= gk(\omega)\tanh(h_0 k(\omega)). 
$$
The shallow water linear reconstruction writes
\begin{equation}\label{rec_linSWb2}
\zeta_{\rm SL}=\zeta_{\rm H}-\frac{h_0}{2g}\left(1-\left(\frac{\delta_m}{h_0}\right)^2\right)\dt^2\zeta_{\rm H},
\end{equation}
and the heuristic formula can be expressed as
\begin{equation}
\label{rec_heurb2}
\zeta_{HE}=\frac{\zeta_{\rm L}} {1+ \frac{ 1}{g}\dt^2 \zeta_{\rm L}}.
\end{equation}


\section{Generalization in the presence of a background current}\label{app_current}

We show here how to generalize the reconstruction formulas derived above when waves propagate over a background constant horizontal current ${\bf U}_0=U_0{\bf e}_x$.  This current cannot be inferred from the pressure measurements and must be obtained from additional velocity measurements. As shown below, the influence of the current can be neglected at the precision of the model if $U_0=O(\eps)$, that is, if the current is smaller by a factor of $\eps$ than the longwave phase velocity.

\subsection{Formulas involving Fourier multiplier in space}\label{app_current1}

 The total velocity field is $\bU=\nabla_{X,z}\Phi+U_0{\bf e}_x$, with $\Phi$ satisfying \eqref{Laplace_ND} and \eqref{BC_fond_ND}, and the Bernoulli and kinematic equations become
$$
\dt \Phi + U_0\dx \Phi+\frac{1}{\eps}z+\frac{\eps}{2}\abs{\nabla\Phi}^2+\frac{\eps}{2\mu}\abs{\dz\Phi}^2 +\frac{1}{2\eps}U_0^2=-\frac{1}{\eps}(P-P_{\rm atm}),
$$
and
$$
\mu \dt \zeta +\mu U_0\dx\zeta =\dz \Phi -\eps\mu \nabla\zeta \cdot \nabla \Phi \quad \mbox{ on }\quad z=\eps\zeta;
$$
where $U_0$ is dimensionalized by $\sqrt{gh_0}$.
Consequently, formula \eqref{formuleexacte_ND} becomes
$$\zeta=\zeta_{\rm H}+(\dt+ U_0\dx) (\Phi_{\rm b}- \psi)+\frac{\eps}{2}\big(\vert \nabla\Phi_{\rm b}\vert^2-\vert \nabla\psi\vert^2\big)+\frac{\eps}{2\mu} (1+ \eps^2\mu\vert \nabla\zeta\vert^2)(\dz \Phi_{\vert_{z=\eps\zeta}})^2 ,
$$
so that the linear reconstruction formula \eqref{rec_lin0} remains valid in the presence of a background current $U_0$,
\begin{align}
\nonumber 
\zeta_{{\rm L},U_0}(t,X)&=\zeta_{{\rm L}}(t,X)\\
\label{rec_lin0_U0}
&=\big[\cosh(\sqrt{\mu}\vert D\vert)\zeta_{\rm H}(t,\cdot)\big](X).
\end{align}
For the nonlinear reconstruction formula, one readily gets that in the presence of a background current, the formula \eqref{rec_quad0} becomes
\begin{align}
\nonumber
\zeta_{{\rm NL},U_0}&=\zeta_{\rm L}-\sqrt{\mu}\sigma (\dt+ U_0\dx) \big(\zeta_{\rm L} (\dt+ U_0\dx) \zeta_{\rm L}\big)\\
\label{rec_quad0_c}
&=\zeta_{\rm L}-\sqrt{\mu}\sigma\big(  (\dt+ U_0\dx) \zeta_{\rm L}\big)^2-\sqrt{\mu}\sigma \zeta_{\rm L} (\dt+  U_0\dx)^2 \zeta_{\rm L};
\end{align}
this formula has the advantage of being Galilean invariant. It can also be seen that if $U_0$ is of order $O(\eps)$ in the frame of reference of the bottom, than the corrections due to the background current are of size $O(\sigma^2)$ and therefore negligible at the precision of the approximation.

\subsection{Formulas involving Fourier multiplier in time}\label{app_current2}

There is  a difference in the way one can pass from a Fourier transform in space to a Fourier transform in time, as in \eqref{rec_lin1}. Indeed, the equation solved at first (linear) order by $\zeta$ is no longer \eqref{eqlin} but
\begin{equation}\label{eqcurrent}
(\dt +U_0\dx)^2u+\lambda(D)^2 u=0,
\end{equation}
with $\lambda(\cdot)$ still given by \eqref{formulalambda}.
In the one-dimensional case ($d=1$), the notion of right-going wave, introduced in Remark \ref{remright}, can be easily extended in the presence of a background current. A solution $u$ to \eqref{eqcurrent} is called {\it right-going} if its Fourier transform is of the form
\begin{equation}\label{RGc}
\widehat{u}(t,\xi)=\widehat{u}(t=0,\xi)\exp\big(-it(U_0 \xi +\lambda(\xi))\big).
\end{equation}
When $U_0\geq  0$ (following current), the function $\xi\in \R\mapsto \lambda(\xi)+U_0\xi$ is one-to-one, and we can define the function $k_{U_0}$ by the dispersion relation
\begin{equation}\label{defk0}
\omega=\lambda(k_{U_0}(\omega))+U_0 k_{U_0}(\omega);
\end{equation}
if $-1 < U_0 < 0$ (opposing current, we only consider here the subcritical case $\vert U_0\vert < 1$), the function $\xi\in \R\mapsto \lambda(\xi)+U_0\xi\in \R$ is no longer one-to-one and we need to introduce the critical wave-number and frequency $k_{\rm crit}\geq 0$ and $\omega_{\rm crit}$ defined as
$$
\lambda'(k_{\rm crit})=-U_0, \qquad \omega_{\rm crit}=\lambda(k_{\rm crit})+U_0 k_{\rm crit},
$$
so that the function $\xi\in\R \mapsto \lambda(\xi)+U_0\xi$ is one-to-one as a mapping $[-k_{\rm crit}, k_{\rm crit}]\to [-\omega_{\rm crit},\omega_{\rm crit}]$ (see Figure \ref{fig_disp}).
\begin{figure}
 \begin{center}
\begin{overpic}[width=0.45\textwidth]{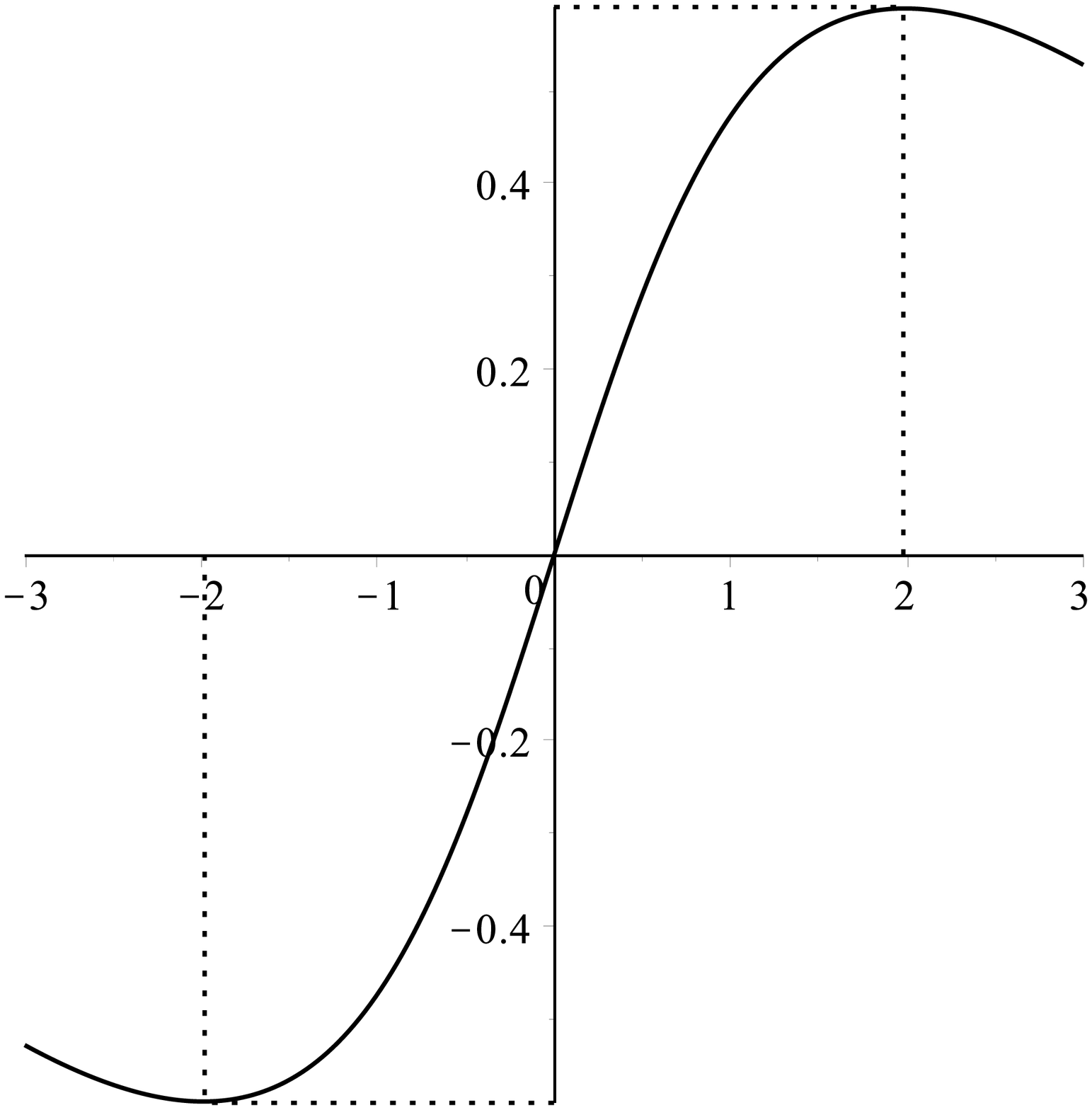}
 \put (20,52) {$\displaystyle -k_{\rm crit}$}
  \put (83,52) {$\displaystyle k_{\rm crit}$}
   \put (52,2) {$\displaystyle -\omega_{\rm crit}$}
      \put (52,95) {$\displaystyle \omega_{\rm crit}$}
\end{overpic}
 \includegraphics[width=0.45\textwidth]{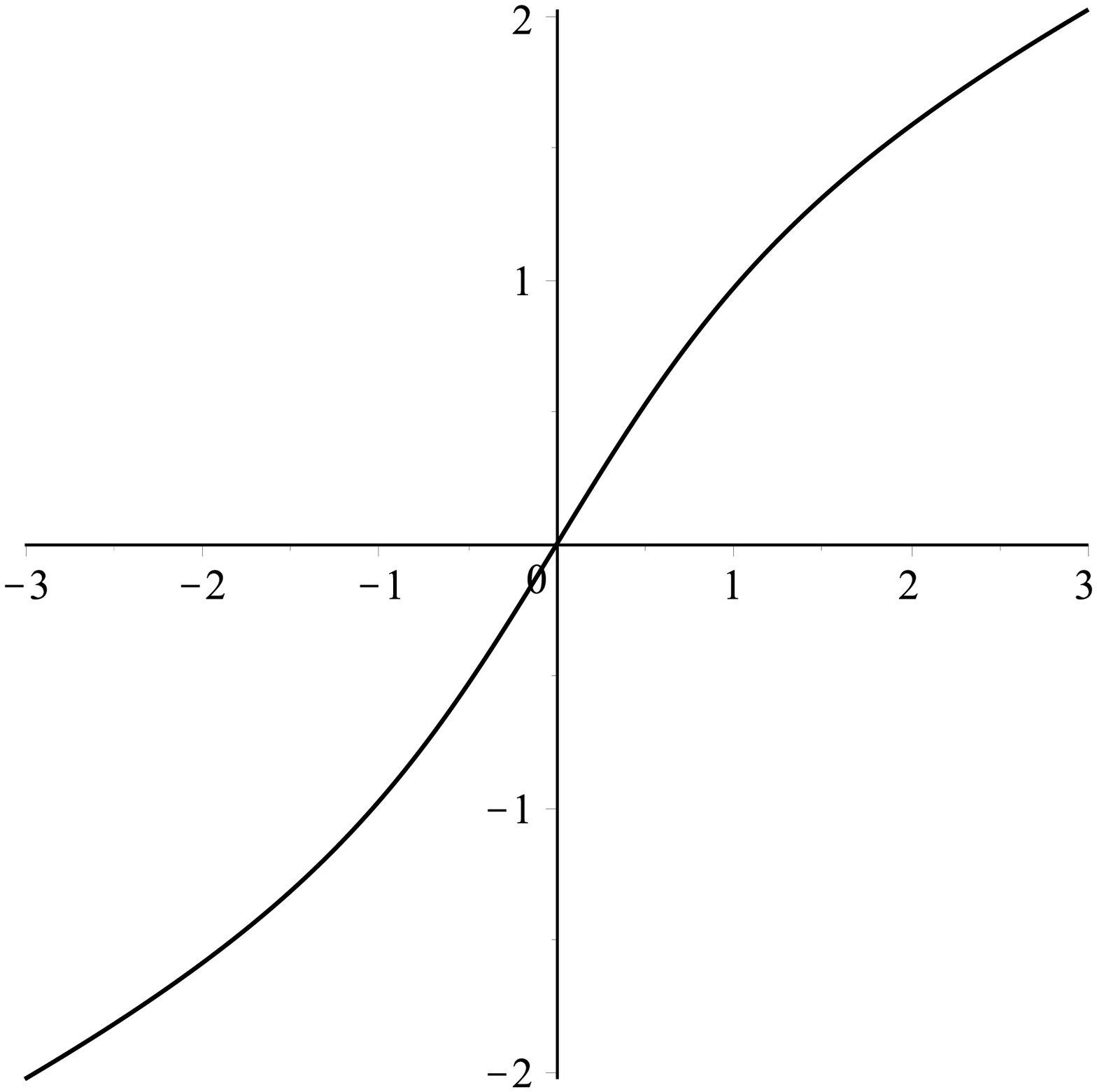}
 \caption{The mapping $\xi\mapsto \lambda(\xi)+U_0 \xi$ for $-1<U_0 < 0$ (left) and $U_0\geq 0$ (right).}
 \label{fig_disp}
 \end{center}
 \end{figure}
It is therefore possible to define $k_{\rm U_0}(\cdot)$ through \eqref{defk0} when $0< U_0\leq 1$, provided that the spatial Fourier transform $\widehat{\zeta}(t,\cdot)$ is supported in $[-k_{\rm crit}, k_{\rm crit}]$, or equivalently if the time Fourier transform $\widetilde{\zeta}(\cdot,x)$ is supported in $[-\omega_{\rm crit},\omega_{\rm crit}]$. We shall therefore make the following assumption
$$
{\bf (H)}\qquad \begin{cases}
U_0\geq 0\\
\mbox{or } -1<U_0 < 0\mbox{ and } \widetilde{\zeta}(\cdot,x)\mbox{ is supported in }[-\omega_{\rm crit},\omega_{\rm crit}]
\end{cases}
$$

\begin{remark}
The linear dispersion for \eqref{eqcurrent} has four solutions for $k$ in general (only two if $U_0 = 0$), when $\omega$ and $U_0$  are given. Among them, two describe counter-propagating long-wavelength waves  while the other two correspond to short-wavelength waves. The above assumption means that we assume that the wave profile we want to reconstruct corresponds to the right-going long-wavelength component (see for instance \cite{EMPPR} for a situation where the counter-propagating waves are relevant).
\end{remark}
We have seen that the linear reconstruction formula \eqref{rec_lin0_U0} is the same as in the absence of any background current. However, when replacing the Fourier multiplier in space by a Fourier multiplier in time, the influence of $U_0$ can be seen. In order to perform such a substitution, we need the following proposition. 
\begin{proposition}\label{prop2}
Let $d=1$, $\lambda$ be given by \eqref{formulalambda}, and let $u$ be a \emph{right-going} solution to \eqref{eqcurrent} that satisfies ${\bf (H)}$. Then, for all Fourier multiplier $f:\R\to \R$, one has
$$
\forall t,x,\qquad [f(D)u(t,\cdot)](x)=[f(-k_{U_0}(D_t))u(\cdot,x)](t),
$$
where $k_{U_0}(\cdot)$ is defined through \eqref{defk0}.
\end{proposition}
Note that contrary to Proposition \ref{prop1}, we assume here that $d=1$ and that the solution is right-going; this allows us to handle Fourier multipliers that are not even functions. 
\begin{proof}
Using the definition \eqref{RGc} of a right-going wave, one easily obtains for the double Fourier transform in space and time that
\begin{align*}
\big({\mathcal F}_{t,x}f(D)u\big)(\omega,\xi)&=f(\xi)\widehat{u}(t=0,\xi)\delta_{\omega=-(U_0 \xi+\lambda(\xi))}\\
&=f(-k_{U_0}(\omega))\widehat{u}(t=0,\xi)\delta_{\omega=-(U_0 \xi+\lambda(\xi))},
\end{align*}
where we used the fact that $k_{U_{0}}(-\omega)=-k_{U_{0}}(\omega)$. Inverting the double Fourier transform then yields the result.
\end{proof}
Using Proposition \ref{prop2}, the linear reconstruction formula \eqref{rec_lin0_U0} implies the following generalization of \eqref{rec_lin1} in the presence of a background current,
\begin{equation}\label{rec_lin1_u0}
\zeta_{{\rm L},U_0}(t,X)=\big[\cosh\big(\sqrt{\mu} k_{U_0}(D_t )\big)\zeta_{\rm H}(\cdot,X)\big](t)\qquad (d=1, \mbox{ right-going}, {\bf (H)}\mbox{ holds}),
\end{equation}
with $k_{U_0}(\cdot)$ given by \eqref{defk0}.

\medbreak
After the estimation of the linear reconstruction $\zeta_{{\rm L},U_0}$, we can use it in the nonlinear reconstruction formula \eqref{rec_quad0_c}. In the presence of a  background opposing current (i.e. $U_0\leq 0$), this formula involves space derivatives of  $\zeta_{{\rm L},U_0}$.
For right-going waves, 
and under the assumption ${\bf (H)}$, we can replace the differential operator $\dx$ by $-ik_{U_0}(D_t)$ using Proposition \ref{prop2} and obtain
\begin{align}
\nonumber  
\zeta_{{\rm NL},U_0}=&\zeta_{{\rm L},U_0}-\sqrt{\mu}\sigma\big(  (\dt-iU_0k_{U_0}(D_t)) \zeta_{{\rm L},U_0}\big)^2\\
&-\sqrt{\mu}\sigma \zeta_{{\rm L},U_0} (\dt-iU_0 k_{U_0}(D_t))^2 \zeta_{{\rm L},U_0}.
\end{align}

%
%

\end{document}